\documentclass[dvipsnames,letterpaper,11pt]{article}

\usepackage[utf8]{inputenc}
\usepackage[english]{babel}

\usepackage{booktabs}
\usepackage{multirow}

\usepackage{xspace}
\usepackage{bm}
\usepackage[tbtags]{amsmath}
\usepackage{amsfonts,amsthm,amssymb}
\usepackage{mathrsfs} 
\usepackage{mathtools}
\usepackage{thmtools}
\usepackage[]{mdframed}

\usepackage{tikz}

\usetikzlibrary{shadows.blur}

\usetikzlibrary{shapes.geometric}

\usepackage{verbatim}
\usepackage{subcaption}

\usepackage[ruled,linesnumbered]{algorithm2e}

\usepackage[margin=1in]{geometry}

\usepackage[dvipsnames,usenames]{xcolor}

\usepackage[pagebackref=true,colorlinks,urlcolor=blue,linkcolor=purple,citecolor=blue,bookmarks,bookmarksopen,bookmarksnumbered]{hyperref}

\usepackage[capitalise,nameinlink]{cleveref}

\usepackage{tcolorbox}

\usepackage{todonotes} 
\usepackage{tablefootnote}

\usepackage{mleftright}

\usepackage{dsfont}
\usepackage[T1]{fontenc}

\newtheorem{theorem}{Theorem}[section]
\newtheorem{lemma}[theorem]{Lemma}

\newtheorem{definition}[theorem]{Definition}

\newtheorem{corollary}[theorem]{Corollary}

\newtheorem{fact}[theorem]{Fact}

\newtheorem{claim}[theorem]{Claim}

\newcommand{\NZerr}{\mathsf{ExtErr}}
\newcommand{\Ext}{\mathsf{Ext}}

\DeclareMathOperator*{\E}{\mathbb{E}}

\newcommand{\PRG}{\mathsf{PRG}}

\numberwithin{equation}{section}

\title{Explicit Min-wise Hash Families with Optimal Size}
\author{Xue Chen\thanks{\tt{xuechen1989@ustc.edu.cn}, University of Science and Technology of China \& Hefei National Laboratory, Hefei 230088, China. Supported by Innovation Program for Quantum Science and Technology 2021ZD0302901, NSFC 62372424, and CCF-HuaweiLK2023006.}
\and Shengtang Huang\thanks{\tt{peanuttang@mail.ustc.edu.cn}, School of the Gifted Young, University of Science and Technology of China.}
\and Xin Li \thanks{\tt{lixints@cs.jhu.edu}, Johns Hopkins University. Supported by NSF CAREER Award CCF-1845349 and NSF Award CCF-2127575.}}
\date{}

\begin{document}

\maketitle

\begin{abstract}
    We study explicit constructions of min-wise hash families and their extension to $k$-min-wise hash families. Informally, a min-wise hash family guarantees that for any fixed subset $X \subseteq [N]$, every element in $X$ has an equal chance to have the smallest value among all elements in $X$; a $k$-min-wise hash family guarantees this for every subset of size $k$ in $X$. Min-wise hash is widely used in many areas of computer science such as sketching \cite{Cohen2016}, web page detection \cite{Henzinger06}, and $\ell_0$ sampling \cite{sampling_survey}. For applications like similarity estimation \cite{Cohen_similarity_est} and rarity estimation \cite{DM_similarity_rarity}, the space complexity of their streaming algorithms is roughly equal to the number of random bits used to construct such families.
    
    The classical works by Indyk \cite{minwise_Indyk} and P{\u{a}}tra{\c{s}}cu and Thorup \cite{PatrascuT16} have shown $\Theta(\log (1/\delta))$-wise independent families give min-wise hash of multiplicative (relative) error $\delta$, resulting in a construction with $\Theta(\log (1/\delta) \log N)$ random bits. While this is optimal for constant errors, it leaves a gap to the existential bound of $O(\log (N / \delta))$ bits whenever $\delta$ is sub-constant, which is needed in several applications. Based on a reduction from pseudorandom generators for combinatorial rectangles by Saks, Srinivasan, Zhou and Zuckerman \cite{SSZZ00}, Gopalan and Yehudayoff \cite{GY15} improved the number of bits to $O(\log N \log \log N)$ for polynomially small errors $\delta$. However, no construction with $O(\log N)$ bits (polynomial size family) and sub-constant error was known before.

    In this work, we continue and extend the study of constructing ($k$-)min-wise hash families from pseudorandomness for combinatorial rectangles and read-once branching programs. Our main result gives the first explicit min-wise hash families that use an optimal (up to constant) number of random bits and achieve a sub-constant (in fact, almost polynomially small) error, specifically, an explicit family of $k$-min-wise hash with $O(k \log N)$ bits and $2^{-O\left(\frac{\log N}{\log \log N}\right)}$ error. This improves all previous results for any $k=\log^{O(1)} N$ under $O(k \log N)$ bits. Our main techniques involve several new ideas to adapt the classical Nisan-Zuckerman pseudorandom generator to fool min-wise hashing with a \emph{multiplicative} error.
\end{abstract}

\thispagestyle{empty}

\newpage
\setcounter{page}{1}

\section{Introduction}\label{sec:intro}

Min-wise hash families play a crucial role in the design of graph algorithms and streaming algorithms. Notable applications include similarity estimation \cite{Cohen_similarity_est}, rarity estimation \cite{DM_similarity_rarity}, data mining \cite{HGKI02_data_mining,Henzinger06}, sketching \cite{Cohen2016}, and $\ell_0$ sampler \cite{sampling_survey,McGregor_survey}. In this work, we study explicit constructions of min-wise hash families with small size. In pseudorandomness, this is equivalent to studying the seed length (number of random bits used) to generate a hash function. In big data algorithms, this is the space complexity of applying min-wise hash families.

Following the standard notation of the min-wise hash \cite{BRODER2000630, minwise_Indyk, minwise_FPS11}, we consider multiplicative (relative) errors with respect to the fair probability in this work. Although this is different from the standard additive errors in pseudorandomness, multiplicative errors are crucial for many algorithmic applications of min-wise hash such as similarity estimation \cite{Cohen_similarity_est} and $\ell_0$ sampling \cite{McGregor_survey}.

\begin{definition}
    Let $a = b \pm \delta$ denote $a \in [b - \delta, b + \delta]$. Then a min-wise hash family $\mathcal{H}=\{h: [N] \rightarrow [M]\}$ of error $\delta$ satisfies that for any $X \subseteq [N]$ and any $y \in X$,
    \begin{align}
    \Pr_{h \sim \mathcal{H}}\left[ h(y) < \min_{x \in X \setminus y} h(x) \right] = \frac{1 \pm \delta}{|X|}.
    \label{eq:def_minwise}
    \end{align}
    
    Moreover, a $k$-min-wise hash family of error $\delta$ satisfies that for any $X \subseteq [N]$ and any $Y \in {X \choose \le k}$,
    \begin{align}
    \Pr_{h \sim \mathcal{H}}\left[ \max_{y \in Y} h(y) < \min_{x \in X \setminus Y} h(x) \right] = \frac{1 \pm \delta}{{|X| \choose |Y|}}.
    \label{eq:def_kminwise}
    \end{align}    

    We sometimes call $\log|\mathcal{H}|$ the \emph{seed length} of the hash family.
\end{definition}

In this work, we will focus on the case of $M=\Omega(N/\delta)$ such that the uniform distribution over all functions from $[N]$ to $[M]$ satisfies \eqref{eq:def_minwise} and \eqref{eq:def_kminwise} \cite{minwise_Indyk,minwise_FPS11}. Specifically, let $U$ be the uniform distribution over all functions $h : [N] \to [M]$. Then $M=\Omega(N/\delta)$ implies
\begin{align}
    \Pr_{h \sim U}\left[ h(y) < \min_{x \in X \setminus y} h(x) \right] = \frac{1 \pm \delta}{|X|} \text{ and } \Pr_{h \sim U}\left[ \max_{y \in Y} h(y) < \min_{x \in X \setminus Y} h(x) \right] = \frac{1 \pm \delta}{{|X| \choose |Y|}}. \label{eq:fair_prob}
\end{align}
This is equivalent to the requirement $|X| = O(\delta N)$ for $M=N$ in previous works \cite{SSZZ00,minwise_Indyk,minwise_FPS11}. Moreover, most applications of min-wise hash could choose the image size $|M|$ to guarantee \eqref{eq:fair_prob}.

Hence, constructing min-wise hash is equivalent to constructing pseudorandom generators (PRGs) with \emph{multiplicative} errors. In this work, we will consider both sides of \eqref{eq:def_kminwise} and \eqref{eq:fair_prob} as the targets of our PRGs.

Although ($k$-)min-wise hash has found a variety of applications in computer science, the primary approaches of explicit constructions are based on $t$-wise independent hash families \cite{minwise_Indyk, minwise_FPS11} and pseudorandom generators for combinatorial rectangles \cite{SSZZ00, GY15}. For min-wise hash, Indyk \cite{minwise_Indyk} showed that any $O(\log (1/\delta))$-wise independent family is a min-wise hash family of error $\delta$. Since $t$-wise independent families need $\Theta(t \log NM)$ random bits, this construction needs $O(\log (1 / \delta) \log NM) $ bits. In fact, P{\u{a}}tra{\c{s}}cu and Thorup \cite{PatrascuT16} showed a matching lower bound: $\Omega(\log (1 / \delta))$-wise independence is necessary to have error $\delta$. In contrast, non-explicitly it is known that one can use $O(\log (NM / \delta))$ random bits to construct min-wise hash families of error $\delta$. Therefore, although the construction in \cite{minwise_Indyk} is optimal for \emph{constant errors}, it fails to be optimal whenever the error is \emph{sub-constant}.

For $k$-min-wise hash, Feigenblat, Porat and Shiftan \cite{minwise_FPS11} showed that $t$-wise independent family is also a $k$-min-wise hash family of error $\delta$ when $t = O(\log (1 / \delta) + k \log \log (1 / \delta))$. In turn, this construction needs $O\big((\log (1 / \delta) + k \log \log (1 / \delta) ) \cdot \log NM \big)$ random bits, which still leaves a gap to the optimal result of $O\big( k \log NM + \log (1/\delta) \big)$ bits for sub-constant $\delta=o(1)$. 

At the same time, Saks, Srinivasan, Zhou and  Zuckerman \cite{SSZZ00} reduced the construction of min-wise hash to pseudorandom generators for combinatorial rectangles of \emph{polynomially small errors}. This reduction translates polynomially small (additive) errors to a multiplicative error like $\delta/|X|$ (described in \hyperref[sec:hash_poly_error]{Appendix A}). Based on this reduction, Gopalan and Yehudayoff \cite{GY15} provided a min-wise hash family of $O(\log NM \log \log NM)$ bits for any polynomially small error. Although this improves the result by Indyk \cite{minwise_Indyk} of $O(\log^2 NM)$ bits when $\delta$ is polynomially small, it does not provide a construction with $O(\log NM)$ bits even when $\delta$ is a constant.

In this work, we study explicit constructions of min-wise hash with small sizes and (almost) polynomially small errors.
Our constructions are well motivated, given that in practice, 
some applications of min-wise hash  require small errors in which the seed length becomes the bottleneck on the space complexity of streaming algorithms. For example, a primary application of min-wise hash is $\ell_0$ sampling in the streaming model. Many graph streaming algorithms need $(k>1)$-min-wise hash with a sub-constant error $\delta=o(1)$ (see Table 1 in \cite{KNPWWY17}) and have space complexity that is equal to the seed length of the min-wise hash family times the number of hashes used \cite{AGM12}. Also, similarity estimation \cite{Cohen_similarity_est,DM_similarity_rarity} applies $k$-min-wise hash of error $\delta$ directly to approximate the Jaccard similarity coefficient $\hat{S}(A,B):= \frac{|A \cap B|}{|A \cup B|}$ between two sets $A$ and $B$ as $(1\pm\delta) \cdot \frac{|A \cap B|}{|A \cup B|} \pm \frac{O(1)}{\delta \cdot k}$, so the space complexity is equal to  the seed length of the $k$-min-wise hash family here.

Furthermore, from a different aspect, given the connection between min-wise hash families and PRGs for combinatorial rectangles shown by Saks, Srinivasan, Zhou and Zuckerman~\cite{SSZZ00}, a natural direction is to apply the results on the long line of research on PRGs for combinatorial rectangles to construct better min-wise hash families.
Constructing pseudorandom generators for combinatorial rectangles have been extensively studied (e.g., \cite{ASWZ96_PRG_rect_comb,LLSZ95,Lu02,GMRTV, GY15} to name a few) because they are related to fundamental problems in theoretical computer science such as derandomizing logspace computation and approximately counting the number of satisfying assignments of a CNF formula. While early works \cite{ASWZ96_PRG_rect_comb,Lu02} in the 90s have already provided PRGs with seed length $O(\log NM)$ and slightly sub-constant errors (e.g., $2^{-\sqrt{\log NM}}$ \cite{ASWZ96_PRG_rect_comb} and $2^{-\log^{2 / 3} NM}$ \cite{Lu02}), no construction of min-wise hash family with $O(\log NM)$ bits and a sub-constant error was known before.

The main bottleneck is that min-wise hash requires a multiplicative error such as $\delta/|X|$. Even for a constant $\delta$, this becomes a polynomially small error like $1/N$ when $|X|=\Omega(N)$. Hence those PRGs for combinatorial rectangles with $O(\log NM)$ seed length do not give a min-wise hash directly. In fact, even after so many years of extensive study on PRGs for combinatorial rectangles, we still don't have explicit constructions of such PRGs with $O(\log NM)$ seed length and $1 / (NM)^{O(1)}$ additive error. Therefore, directly applying these PRGs is not enough to get a min-wise hash family with seed length $O(\log NM)$. To address this barrier, in this work we provide several new ideas to extend the Nisan-Zuckerman PRG \cite{NZ96} and construct min-wise hash families with $O(\log NM)$ seed length and almost polynomially small errors.

\subsection{Our Results}

Our main results are an explicit construction of min-wise hash families with seed length $O(\log N)$ and almost polynomially small errors and its generalization to $k$-min-wise hash. For ease of exposition, we assume $M=(N/\delta)^{O(1)}$.

\begin{theorem}\label{thm:inform_min_wise_log_bits}
    Given any $N$, there exists an explicit family of min-wise hash of $O(\log N)$ bits and (multiplicative) error $\delta=2^{-O\left(\frac{\log N}{\log \log N}\right)}$.
\end{theorem}

We remark that the seed length of our min-wise hash is optimal up to constants and the error is almost polynomial up to a $\log \log N$ factor in the exponent. Hence \Cref{thm:inform_min_wise_log_bits} improves previous results \cite{SSZZ00,minwise_Indyk,minwise_FPS11, GY15} for seed length $O(\log N)$. In fact, this is the first construction of min-wise hash family with optimal seed length and sub-constant error.

Next we state its generalization to $k$-min-wise hash.

\begin{theorem}\label{thm:inform_kminwise} 
    Given any $k = \log^{O(1)} N$, there exists an explicit $k$-min-wise hash family of $O(k \log N)$ bits and  (multiplicative) error $\delta=2^{-O\left(\frac{\log N}{\log \log N}\right)}$.    
\end{theorem}

Again, this is the first construction of $k$-min-wise hash with optimal seed length and sub-constant error. One remark is that $k$-min-wise hash requires $\Omega(k \log N)$ bits even for a constant error. This is because the fair probability could be as small as $1/{N \choose k}$ in Definition \eqref{eq:def_kminwise}. Also, as a direct application, our constructions give the optimal space complexity for many applications including similarity estimation and rarity estimation \cite{Cohen_similarity_est,DM_similarity_rarity} whenever $k = \log^{O(1)} N$ and $\delta=2^{-O\left(\frac{\log N}{\log \log N}\right)}$.

\subsection{Technique Overview}\label{sec:overview}

First of all, the connection shown in \cite{SSZZ00} is not enough to directly use known PRGs for combinatorial rectangles to construct min-wise hash families with $O(\log N)$ bits. This is because one remarkable feature of min-wise hash is that the error is \emph{multiplicative} with respect to the fair probability $1/|X|$, which could be as small as $1/N$. On the other hand, standard pseudorandom generators only consider additive errors, and constructing $O(\log NM)$ seed length PRGs fooling combinatorial rectangles with error $1/N$ is still a big open problem. More broadly, the long line of research on classical PRGs for read-once branching programs (ROBPs) (\cite{Nisan92,INW94,NZ96,BRRY,BV,KNP11,FK18,MRT19} to name a few) does not give a multiplicative error with $O(\log N)$ bits of seed. 

To overcome this barrier, our main technical contribution is to extend the Nisan-Zuckerman PRG framework \cite{NZ96} to achieve a small multiplicative error. For convenience, we use $\max h(S):=\max_{x \in S} h(x)$ and $\min h(S):=\min_{x \in S} h(x)$ in the rest of this work. To illustrate our ideas, let us consider how to fool $\Pr_{h \sim U}[h(y)<\min h(X \setminus y)]$ for a sub-constant error with $O(\log N)$ bits.
    
It would be more convenient to enumerate $\theta:=h(y)$ and decompose 
\begin{equation}\label{eq:intro_decomp}
\Pr_{h \sim U}[ h(y)<\min h(X \setminus y) ] = \sum_{\theta \in [M]} \Pr_{h \sim U}[ h(y)=\theta \wedge \min h(X \setminus y) > \theta ], 
\end{equation} instead of analyzing $\Pr_{h \sim U}[h(y) < \min h(X \setminus y)]$ itself (because $\Pr_{h \sim U}[h(y)<h(x_1)]$ and $\Pr_{h \sim U}[h(y)<h(x_2)]$ are correlated). Since the first event  $\Pr_U[h(y)=\theta]=1/M$ in \eqref{eq:intro_decomp}, our first goal is to fool $\Pr_{h \sim U}[ h(y)=\theta \wedge \min h(X \setminus y) > \theta ]$ in \eqref{eq:intro_decomp} with a multiplicative error like $\delta \cdot \Pr_{h \sim U}[h(y)=\theta]=\delta/M$ for $O(\log N)$ bits (assuming $M=N^{O(1)}$). We remark that this is a polynomially small additive error. While $\mathbf{1}(h(y)=\theta \wedge \min h(X \setminus y) > \theta)$ ($\mathbf{1}$ denotes the indicator function) is a combinatorial rectangle and a simple ROBP of width $2$, no known PRGs for combinatorial rectangles or ROBPs can fool it with additive error $\delta/M$ given $O(\log NM)$ bits of seed. 
 
Since most PRGs for combinatorial rectangles and ROBPs are based on Nisan's PRG \cite{Nisan92} (or its extension to the INW PRG \cite{INW94}), a first attempt would be to modify these PRGs. However, the random event $(h(y)=\theta)$ already takes $\log_2 M$ bits. Since $y$ could be any element, it is unclear how to revise these PRGs to replenish so many random bits just for one step (of $h(y)$).
 
Another candidate is the Nisan-Zuckerman PRG \cite{NZ96}. Recall the basic construction of the Nisan-Zuckerman PRG: It prepares a random source $w$ of length $C \log N$ and $\ell=\log^c N$ seeds of length $\frac{\log N}{\ell}$ (for two constants $C>1$ and $c<1$); then it applies an extractor $\Ext:\{0,1\}^{C \log N} \times \{0,1\}^{(\log N)/\ell} \rightarrow \{0,1\}^{\frac{C}{3} \log N}$ (see \Cref{def:extractor}) to obtain $\ell$ outputs $\Ext(w,s_1), \Ext(w,s_2), \cdots ,\Ext(w,s_\ell)$. The analysis relies on the fact that a ROBP (or a small-space algorithm) cannot record too much information of $w$, therefore each $\Ext(w,s_i)$ is close to an independent and uniform random string. 
 
While this PRG only outputs $\frac{C}{3} \log N \cdot \ell=O(\log^{1 + c} N)$ bits, one could stretch it to a vector in $[M]^N$ via balls-into-bins: first hashing these $N$ variables into $\ell$ buckets and then using $\Ext(w,s_i)$ to generate a $C/3$-wise independent function for every bucket. 
 
However, due to the limited length of $s_i$, the error of each $\Ext(w,s_i)$ is $2^{-O\left(\frac{\log N}{\ell}\right)}=N^{-o(1)}$, which is too large compared to $\Pr[h(y)=\theta]=1/M$. 
    
To address this issue, our starting point is that the Nisan-Zuckerman PRG can fool ROBPs with any input order. More attractively, we can choose the input order in our favor. We provide two constructions that explore this advantage in two different ways. Both can fool $\Pr_{h \sim U}[h(y)=\theta \wedge \min h(X \setminus y) > \theta]$ with an error $\delta/M$ for $\delta=2^{-O\left(\frac{\log N}{\log \log N}\right)}$.

\paragraph{Approach 1.} We consider a special type of extractors, which provides a strong guarantee when the source is uniform. Observe that it never hurts to put $h(y)$ as the first input of the ROBP under the Nisan-Zuckerman PRG. Suppose $y$ is in bucket $j \in [\ell]$ such that the value $h(y)$ is generated by $\Ext(w,s_j)$. Our observation here is that as the first input, the random source $w$ is uniform at the beginning, hence one could expect stronger properties for $\Ext(w,s_j)$ than for other $\Ext(w,s_i)$ where $i \neq j$. In particular, we build an extractor such that $\Ext(U_n,s)$ is uniform for a uniform source and any fixed seed. The construction is based on the sequence of works in pseudorandomness that designed \emph{linear seeded} randomness extractors \cite{NZ96,Trevisan,GUV}. 

Back to the construction of min-wise hash, this guarantees $\Ext(w,s_j)$ is uniform (without any error) such that it has a fair probability of $\Pr[h(y)=\theta]$. Hence the Nisan-Zuckerman PRG has a multiplicative error with respect to $\Pr[h(y)=\theta]=1/M$ when it applies the extractor described above. 

However, plugging this into \eqref{eq:intro_decomp} only gives an error like $1/\log^{O(1)} N$ because $\ell \le \log N$ and it only provides constant-wise independence in each bucket. To reduce the error to be as small as possible, we apply the PRG for combinatorial rectangles to generate $\ell=2^{\frac{\log N}{\log \log N}}$ seeds and use $\Ext(w,s_j)$ to provide both $t$-wise independence and pseudorandomness against combinatorial rectangles. We describe this construction in \Cref{sec:smaller_error}.

\paragraph{Approach 2.} Next we consider how to build a hash family $\mathcal{H}$ fooling \eqref{eq:intro_decomp} like a conditional event
\begin{equation}\label{eq:condition_event}
\Pr_{h \sim \mathcal{H}}[h(y)=\theta] \cdot \Pr_{h \sim \mathcal{H}} [ \min h(X \setminus y) > \theta \mid h(y) = \theta ] = \Pr_{h \sim U}[h(y) = \theta] \cdot \left( \Pr_{h \sim U}[\min h(X \setminus y) > \theta] \pm \delta \right)
\end{equation}
with a multiplicative error $\delta$. On the one hand, $O(1)$-wise independent family would guarantee $\Pr_{h \sim \mathcal{H}}[h(y)=\theta]=\Pr_{h \sim U}[h(y)=\theta]$. On the other hand, $\mathbf{1}(\min h(X \setminus y) > \theta)$ is a combinatorial rectangle where several PRGs \cite{ASWZ96_PRG_rect_comb,Lu02,GMRTV, GY15} can fool it with a small additive error using $\widetilde{O}(\log N)$ bits of seed. This suggests us to consider the direct sum of a $t$-wise independent family and a PRG for combinatorial rectangles (see definition in \Cref{sec:preli}). However, it is unclear how to argue that this sum fools the conditional event $\Pr_{h \sim \mathcal{H}} [ \min h(X \setminus y) > \theta \mid h(y) = \theta ]$. 

Our key observation here is that the sum of a $t$-wise independent family and the Nisan-Zuckerman PRG can actually fool it! Roughly, this is because as before, one can put the output $\Ext(w,s_j)$ generating $h(y)$ at the beginning and fix it. Then we can use $t$-wise independence to argue about $\Pr[h(y)=\theta]$ for this bucket and we fix a specific $t$-wise independent function satisfying $h(y)=\theta$. These two steps only reduce the min-entropy of $w$ by a constant fraction, therefore the Nisan-Zuckerman framework still works. We describe how to build $k$-min-wise hash based on this approach in \Cref{sec:k_min_wise}.

\subsection{Discussion}

In this work, we study explicit constructions of small size min-wise hash functions. Although Saks, Srinivasan, Zhou and Zuckerman \cite{SSZZ00} have reduced the construction of min-wise hash to PRGs for combinatorial rectangles, previous works do not provide any explicit family of sub-constant (multiplicative) errors with $O(\log N)$ bits. 

Our main technical contribution is to construct $k$-min-wise hash of $O(k \log N)$ bits and almost-polynomial $2^{-O\left( \frac{\log N}{\log \log N} \right)}$-error via the Nisan-Zuckerman framework for any $k = \log^{O(1)} N$. Our results extend the Nisan-Zuckerman framework in several aspects. For example, our construction shows that one could guarantee one output of those extractions is uniform (without any error). We also show that the direct sum of the Nisan-Zuckerman PRG with other PRGs could fool conditional events and provide multiplicative errors with respect to small probability events. 

We list several open questions here.
\begin{enumerate}
    \item For $k$-min-wise hash with a (relative) large $k$ like $\log N$, can we have constructions with $O(k \log N)$ bits and polynomially small (multiplicative) error? As mentioned earlier, $k$-min-wise hash needs at least $\Omega(k \log N)$ bits, which is $\Omega(\log^2 N)$ when $k=\Omega(\log N)$. At the same time, there are many PRGs for combinatorial rectangles with polynomially small (additive) errors and $O(\log^2 N)$ seed length.

    \item It is interesting to investigate PRGs fooling conditional events like \eqref{eq:condition_event}. For example, can we show the direct sum of a $t$-wise independent function and the Nisan PRG has a small multiplicative error for \eqref{eq:condition_event}?

    \item The fact that the Nisan-Zuckerman PRG can fool any input order has been used to fool formulas and general branching programs \cite{IMZ}. Can we find more applications of this powerful framework?
\end{enumerate}

\subsection{Related Works}

Min-wise hash was introduced and investigated by Broder, Charikar, Frieze and Mitzenmacher \cite{BRODER2000630} where the first definition set the probability to be exact $1/|X|$. This is equivalent to requiring that each function in the family is a permutation when $M=N$. Such a family is also called min-wise permutation.
However, they showed a lower bound $\Omega(N)$ on the number of bits for the exact probability, and suggested to consider min-wise hash with multiplicative (relative) error for applications like similarity estimation and duplicate detection.

Later on, Indyk showed the first construction that $O(\log (1/\delta))$-wise independent families are min-wise hashing of error $\delta$. A matching lower bound on $t$-wise independent family was shown by P{\u{a}}tra{\c{s}}cu and Thorup \cite{PatrascuT16} later. 

For polynomially small errors, Saks, Srinivasan, Zhou and Zuckerman \cite{SSZZ00} provided a construction of $O(\log^{3 / 2} N)$ bits based on the PRGs for combinatorial rectangles; this was improved to $O(\log N \log \log N)$ by Gopalan and Yehudayoff \cite{GY15}. This work \cite{GY15} is still the state-of-the-art of both PRGs for combinatorial rectangles and min-wise hash given \emph{polynomially small errors}.

While one could run $k$ parallel min-wise hash to sample $k$ elements with replacement, $k$-min-wise looks for a sampling without replacement. This turns out to be more accurate in practice \cite{Cohen_similarity_est,DM_similarity_rarity}. Feigenblat, Porat and Shiftan \cite{minwise_FPS11} showed that $O(\log (1 / \delta) + k \log \log (1 / \delta))$-wise independent families are $k$-min-wise hashing of error $\delta$. Based on PRGs for combinatorial rectangles, Gopalan and Yehudayoff \cite{GY15} constructed $k$-min-wise hash of $O\big(k \log N \cdot \log (k \log N) \big)$ bits for polynomially small errors. 

Min-wise hash families and combinatorial rectangles are subclasses of read-once branching programs. So the classical PRG by Nisan \cite{Nisan92} of $O(\log^2 N)$ bits implies a min-wise hash family of the same seed length. A long line of research has studied the effects and limitations of Nisan's PRG (to name a few \cite{INW94,BRRY,BV,De,KNP11}). However, there has been little quantitative progress on the improvement of Nisan's PRG. One exception is the construction of PRGs for combinatorial rectangles, where subsequent works \cite{ASWZ96_PRG_rect_comb,Lu02} have reduced the seed length to $O(\log N)$ and achieved smaller errors.

The Nisan-Zuckerman PRG \cite{NZ96} provides another method to derandomize ROBPs of seed length $O(\log N)$. While its output length is just $\log^{O(1)} N$, it can fool ROBPs of any input order. Impagliazzo, Meka and Zuckerman \cite{IMZ} have used this property to fool general formulas and branching programs.

Another powerful paradigm to design PRGs for ROBPs is via milder restrictions. Several beautiful applications are the PRGs for combinatorial rectangles and read-once CNFs \cite{GMRTV, GKM15, GY15}, the PRGs for ROBPs with an arbitrary input order \cite{FK18}, and the PRGs for ROBPs of width 3 \cite{MRT19}.

\subsection{Paper Organization}

In \Cref{sec:preli}, we describe 
basic notations, definitions, and useful theorems from previous works. In \Cref{sec:smaller_error}, we show an explicit construction of min-wise hash based on the first approach described in \Cref{sec:overview}, which proves \Cref{thm:inform_min_wise_log_bits}. In \Cref{sec:k_min_wise}, we show another construction of $k$-min-wise hash based on the second approach described in \Cref{sec:overview}, which proves \Cref{thm:inform_kminwise}. In \Cref{sec:extractor}, we show the extractor whose output is uniform when the input source is uniform.
\section{Preliminaries}\label{sec:preli}

\paragraph{Notations.} For three real variables $a$, $b$ and $\delta$, $a = b \pm \delta$ means the error between $a$ and $b$ is $|a - b| \le \delta$. 

Let $[n]$ denote $\{1,2,\ldots,n\}$ and ${n \choose k}$ denote the binomial coefficient. For a subset $X$, we use ${X \choose k}$ to denote the family of all subsets with size $k$. For a vector or a string, we use $|\cdot|$ to denote its dimension or length.

In this work, for a function $h:[N] \rightarrow [M]$, we view it as a vector in $[M]^N$ and vice versa. For a subset $S \subseteq [N]$, let $h(S)$ denote the sub-vector in $S$. Then we use $\max h(S)$ to denote $\max_{x \in S} h(x)$, $\min h(S)$ to denote $\min_{x \in S} h(x)$ and $h(S) = \theta$ to denote the event $(h(x) = \theta,\ \forall x \in S)$.

For two functions $f,g:[N] \rightarrow [M]
$, we use $f + g$ to denote their \emph{direct sum} on every entry $x$ in $[M]$, i.e., $(f + g)(x) = (f(x) + g(x) - 1) \mod M + 1$. Similarly, for two vectors $f$ and $g \in [M]^N$, $f + g$ denotes the corresponding vector.

\paragraph{Combinatorial rectangles and read-once branching programs.} For an event $A$, let $\mathbf{1}(A)$ denote its indicator function. Given the alphabet $[M]$ and $N$ subsets $S_1,\ldots,S_N \subseteq [M]$, its combinatorial rectangle is the function $f: [M]^N \rightarrow \{0,1\}$ defined as the product of $N$ independent events $x_1 \in S_1, \ldots, x_N \in S_N$: $f(x):=\prod_{i=1}^N \mathbf{1}(x_i \in S_i)$. 

Equivalently, it is a function $f:[M]^N \rightarrow \{0,1\}$ defined as $f(v)=\prod_{i=1}^N f_i(v_i)$ by $N$ arbitrary functions $f_1,\ldots,f_N:[M] \rightarrow \{0,1\}$. 

Combinatorial rectangles are a special type of read-once branching programs. 

\begin{definition}[Read-once branching program]
    An width-$w$ length-$n$ read-once branching program on alphabet $\Gamma$ is a layered directed graph $M$ with $n + 1$ layers and $w$ vertices per layer with the following properties.
    \begin{itemize}
        \item The first layer has a single start node and the vertices in the last layer are labeled by $0$ or $1$.

        \item Each vertex $v$ in layer $i\ (0 \le i < n)$ has $|\Gamma|$ edges to layer $i + 1$, each labeled with an element in $\Gamma$.
    \end{itemize}
\end{definition}

A graph $M$ as above naturally defines a function $M : \Gamma^n \to \{0, 1\}$, where on input $(x_1, \ldots, x_n) \in \Gamma^n$, one traverses the edges of the graph according to the labels and outputs the label of the final vertex reached.

\paragraph{Pseudorandomness.} For a fixed domain like $[n]$ or $[M]^N$, we use $U$ to denote the uniform distribution on this domain. Moreover, we use $U_n$ to denote the uniform distribution in $\{0,1\}^n$.

For two distributions $D$ and $D'$ in the domain, we use $D \approx_{\varepsilon} D'$ to indicate that their statistical distance is at most $\varepsilon$ and call this fact $D$ is $\varepsilon$-close to $D'$.

\begin{definition}[Pseudorandom generator]
    Given a fixed domain $D$ and a family of functions from $D$ to $\{0, 1\}$, a pseudorandom generator (PRG) $G:\{0,1\}^{\ell} \rightarrow D$ $\varepsilon$-fools this family $\mathcal{F}$ if
    \[
    \forall f \in \mathcal{F},\ \E_{s \sim U_\ell}[f(G(s))]=\E_{x \sim U}[f(x)] \pm \varepsilon.
    \]
    We call $\ell$ the seed length of $G$ and $\varepsilon$ its error.
\end{definition}

One basic component to construct pseudorandom generators is $t$-wise independent family(function).

\begin{definition}
    We say a distribution $D$ is $t$-wise independent in $[M]^N$ if for any $k$ distinct positions $i_1,\ldots,i_k$ in $[M]$, the marginal distribution $D(x_{i_1},\ldots,x_{i_k})$ is uniform on $[M]^k$.
\end{definition}

Explicit constructions of $t$-wise independent family of seed length $O(t \log NM)$ are well known. We state two useful bounds on $t$-wise independent random variables  \cite{minwise_Indyk,minwise_FPS11}. In the rest of this work, we use $\E_{X: t\text{-wise}}[f(X)]$ denote the expectation of $f(X)$ when $X$ is $t$-wise independent.

\begin{lemma}\label{lem:t_wise_ind}
    Let $\sigma : [N] \to [M]$ a function sampled from $t$-wise-independence, and let $B \subseteq [N]$ be a subset with $b := |B|$. Then, for $\Pr_\sigma[\min \sigma(B) > \theta]$, the following two estimates hold:
    \begin{enumerate}
        \item $\Pr_\sigma[\min \sigma(B) > \theta] = (1 - \theta / M)^b \pm \left( b \cdot \frac{\theta}{M} \right)^t/t!$;
        
        \item $\Pr_\sigma[\min \sigma(B) > \theta] \le \left( \frac{C_t \cdot t}{b \cdot \theta / M} \right)^{t / 2}$, where $C_t$ is a universal constant.
    \end{enumerate}
\end{lemma}

\paragraph{Randomness Extractor.} Our construction will be based on randomness extractors. The min-entropy of a random source $X$ is defined as $H_\infty(X) = \log \max_{x} 1/\Pr[X=x]$.

\begin{definition}\label{def:extractor}
    $\Ext:\{0,1\}^n \times \{0,1\}^d \rightarrow \{0,1\}^m$ is a $(k,\varepsilon)$-randomness extractor if for any source $X$ of min-entropy $k$, $\Ext(X,U_d) \approx_{\varepsilon} U_m$.
\end{definition}

Our PRG needs a randomness extractors with an extra property: $\Ext(U_n, s) = U_m$, whose proof is deferred to \Cref{sec:extractor}.

\begin{lemma}\label{lem:extractor_unif}
    Given any $n$ and $k < n$, for any error $\varepsilon$, there exists a randomness extractor $\Ext : \{0, 1\}^n \times \{0, 1\}^d \to \{0, 1\}^m$ with $m = k/ 2$ and $d = O(\log (n / \varepsilon))$. Moreover, $\Ext$ satisfies an extra property: $\Ext(U_n, s) = U_m$ for any fixed seed $s$. 
\end{lemma}

\paragraph{Pseudorandomness for combinatorial rectangles.} PRGs for combinatorial rectangles have been extensively studied in the past few decades \cite{ASWZ96_PRG_rect_comb,LLSZ95,Lu02,GMRTV,GY15}. Our constructicons of min-wise hash are based on these PRGs and their techniques. We state the state-of-the-art here by Gopalan and Yehudayoff, i.e., Theorem 1.9 in \cite{GY15}.

\begin{theorem}\label{thm:PRG_comb_rect}
    For any error $\varepsilon$, dimension $N$, and alphabet $[M]$, there exists a PRG of seed length $O\left(\log \left( \frac{M \log N}{\varepsilon} \right) \cdot \log \log (M/\varepsilon)\right)$ that fools combinatorial rectangles in $[M]^N$ within additive error $\varepsilon$.
\end{theorem}

Based on the reduction by Saks, Srinivasan, Zhou and Zuckerman \cite{SSZZ00}, this provides a construction of min-wise hash of $O(\log NM \log \log NM)$ bits and polynomially small errors. For completeness, we provide this reduction in \hyperref[sec:hash_poly_error]{Appendix A}.

A direct corollary of \Cref{thm:PRG_comb_rect} provides a PRG of seed length $O(\log NM)$ and almost polynomially small error $2^{-O\left(\frac{\log NM}{\log \log NM}\right)}$. This is based on reductions between PRGs by Lu \cite{Lu02}. Specifically, Lu constructed a PRG for combinatorial rectangles via a sequence of reductions. In particular, Lemma 3 in \cite{Lu02} uses $O(\log (N M / \varepsilon))$ random bits to reduce the original problem to a problem in $[1/\varepsilon^{O(1)}]^{1/\varepsilon^{O(1)}}$ within error $\varepsilon$. Since the PRG in \Cref{thm:PRG_comb_rect} fools combinatorial rectangles in $[1/\varepsilon^{O(1)}]^{1/\varepsilon^{O(1)}}$ within error $\varepsilon$ and $O(\log (1/\varepsilon) \log \log (1/\varepsilon))$ bits. After setting $\varepsilon=2^{-C\cdot \frac{\log NM}{\log \log NM}}$, this gives a PRG of seed length $O(\log NM)$.

\begin{corollary}\label{cor:almost_poly_small_error}
    For any constant $C$ and error $\varepsilon=2^{-C \cdot \frac{\log NM}{\log \log NM}}$, there exists an explicit PRG of seed length $O_C(\log NM)$ that fools combinatorial rectangles within additive error $\varepsilon$.
\end{corollary}
\section{Min-wise Hash of Polynomial Size}\label{sec:smaller_error}

The goal of this section is to provide a construction of explicit min-wise hash family with polynomial size and \emph{nearly} polynomial small error. 

Recall the definition that $\mathcal{H} = \{h_i : [N] \to [M]\}$ is a min-wise hash family of error $\delta$, if for any $X \subseteq [N]$ and any $y \in X$, $\Pr_{h \sim \mathcal{H}}[h(y) < \min h(X \setminus y)] = \frac{1 \pm \delta}{|X|}$. Since the multiplicative error $\delta$ would be $\Omega(1/N)$, we assume the size of the alphabet $M = (N / \delta)^{O(1)} = N^{O(1)}$ for convenience.

\begin{theorem}\label{thm:min_wise_log_seed}
    Given any $N$ and any constant $C$, there exists an explicit min-wise hash family whose seed length is $O_{C}(\log N)$ and multiplicative error is $\delta = 2^{-C \cdot \frac{\log N}{\log \log N}}$.
\end{theorem}

Besides the Nisan-Zuckerman PRG, our construction uses several extra ingredients. 
The first one is to use the extractor in \Cref{lem:extractor_unif} with $\Ext(U_n, s) = U_m$ and an asymptotic optimal error.  The second idea is to generate random seeds in the Nisan-Zuckerman PRG by the PRG of combinaroial rectangles in \Cref{thm:PRG_comb_rect}, which is motivated by a domain reduction of Lu \cite{Lu02}. 

Now we describe the construction of our hash family $\mathcal{H} = \{h_i : [N] \to [M]\}$.

\begin{tcolorbox}
    \textbf{Min-wise Hash Family: dimension $N$, error $2^{-C \cdot \frac{\log N}{\log \log N}}$, and alphabet $N^{O(1)}$}
    \begin{enumerate}
        \item Set $\ell := 2^t$ for $t := \frac{\log N}{\log \log N}$ and pick large constants $C_e$, $C_s$ and $C_g$ such that $C_e > C_s > C_g > C$.
        
        \item Sample a $C_g$-wise independent function $g : [N] \to [\ell]$ as the allocation of $[N]$ into $\ell$ buckets.
    
        \item Apply $\textsf{PRG}_1$ in \Cref{thm:PRG_comb_rect} fooling combinatorial rectangles of seed length $O(\log N)$ to generate $\ell$ pseudorandom seeds $s_1, \ldots, s_\ell$ in $\{0, 1\}^{C_e \cdot t}$ with an additive error $2^{- (C_s+C_e) \cdot t - 2}$.    
        
        \item Sample a random source $w \sim \{0,1\}^{C_e \cdot \log N}$ and apply the extractor in \Cref{lem:extractor_unif}, $\Ext : \{0, 1\}^{C_e \cdot \log N} \times \{0, 1\}^{C_e \cdot t} \to \{0,1\}^{0.3 C_e \cdot \log N}$ of min-entropy $0.6 C_e \cdot \log N$ and error $\NZerr = 2^{-C_s \cdot t}$, to $w$ and $s_1, \ldots, s_\ell$. For every $i \in [\ell]$, let $z_i := \Ext(w, s_i)$.
        
        \item Define a hash family $\mathcal{G} = \{ G_1, \ldots, G_{N^{0.3 C_e}} : [N] \to [M] \}$ of size $N^{0.3 C_e}$ to be the direct sum of a $(0.1 C_s + 1)$-wise independent function in $[M]^N$ and $\textsf{PRG}_2$ of error $2^{-C_s \cdot t}$ for combinatorial rectangles in $[M]^N$ from \Cref{cor:almost_poly_small_error}.
        
        \item Use $z_i$ to pick a function in $\mathcal{G}$ and denote it by $\sigma_i := G_{z_i}$.
        
        \item Finally, let hash function $h(x) = \sigma_{g(x)}(x)$ for every $x \in [N]$.
    \end{enumerate}
\end{tcolorbox}

We finish the proof of \Cref{thm:min_wise_log_seed} in this section. Firstly, we have the following properties from the allocation function $g$. For ease of exposition, let $j := g(y)$ and $B_i:=\{x \in X \setminus y : g(x) = i\}$ in the rest of this section.

\begin{lemma}\label{lem:allocation_small_error}
    Let $C_g$ be a sufficiently large constant compared to $C$. Then $C_g$-wise independent function $g:[N] \rightarrow [\ell]$ guarantees that (recall $\ell=2^{\frac{\log N}{\log \log N}}$):
    \begin{enumerate}
        \item When $|X| \le \ell^{0.9}$, with probability $1 - 1 / \ell^{3 C}$, $|B_i| \le C_g$ for all $i \in [\ell]$.

        \item When $|X| \in (\ell^{0.9}, \ell^{1.1})$, with probability $1 - 1 / \ell^{3 C}$, all buckets have $|B_i| \le 2 \ell^{0.1}$.
        
        \item When $|X| \ge \ell^{1.1}$, with probability $1-|X|^{-3 C}$, all buckets satisfy $|B_i| = (1 \pm 0.1) \cdot |X| / \ell$.
    \end{enumerate}
\end{lemma}

The key point is that since $\ell = 2^t = 2^{\frac{\log N}{\log \log N}}$, the failure probability of each case is small enough compared to the error $\delta / |X|=2^{-C \cdot \frac{\log N}{\log \log N}}/|X|$. Thus, we can assume that all properties in \Cref{lem:allocation_small_error} hold.

To calculate $\Pr_{h \sim \mathcal{H}}[h(y)<\min h(X \setminus y)]$, we express this as
\begin{align}
      & \Pr_{h \sim \mathcal{H}}[h(y) < \min h(X \setminus y)] = \sum_{\theta \in [M]} \Pr_{h \sim \mathcal{H}} [h(y)=\theta \wedge \min h(X \setminus y)>\theta] \notag \\
    = & \sum_{\theta \in [M]} \Pr_{w \sim U, (s_1, \ldots, s_\ell) \sim \PRG_1} \left[ (\sigma_j(y) = \theta \wedge \min \sigma_j(B_j) > \theta) \wedge (\min \sigma_i(B_i) > \theta,\ \forall i \ne j)  \right], \label{eq:minwise_over_theta}
\end{align}

\noindent where function $\sigma_i := G_{z_i}$ for $z_i := \Ext(w, s_i)$. Our analysis will bound each term of \eqref{eq:minwise_over_theta}. 

Our second step is to bound the multiplicative error when the seeds $s_1, \ldots, s_\ell$ of extractors are sampled independently and uniformly from $\{0, 1\}^{C_e \cdot t}$ like the Nisan-Zuckerman PRG. 

\begin{lemma}\label{lem:independent_seeds}
    Let $\mathcal{H}'$ be the hash function family after replacing $s_1, \ldots, s_\ell$ by independent random samples in $\{0, 1\}^{C_e \cdot t}$, instead of applying $\PRG_1$. The (multiplicative) error of $\mathcal{H'}$ is at most $2^{-2 C \cdot t}$:
    \begin{align}            
        \Pr_{h' \sim \mathcal{H}'}[h'(y) < \min h'(X \setminus y)] & = \sum_{\theta \in [M]} \Pr_{h' \sim \mathcal{H}'}\left[ h'(y) = \theta \wedge \min  h'(X \setminus y)>\theta \right] \notag \\
        & = ( 1 \pm 2^{-2 C \cdot t} ) \cdot \Pr_{\sigma \sim U}[\sigma(y) < \min \sigma(X \setminus y)]
        \label{eq:independent_seeds}.
    \end{align}
\end{lemma}

Next we consider the error when the hash family $\mathcal{H}$ uses correlated seeds $s_1, \ldots, s_\ell$ generated by $\PRG_1$ for combinatorial rectangles. We bound the error between $h$ and $h'$ as follows.

\begin{claim}\label{clm:dependent_seeds}
    For any fixed $g$, let $j:=g(y)$ and $B_j:=\{x \in X\setminus y:g(x)=g(y)\}$. Then we define $\mathcal{H}_g$ to be the hash family with this fixed allocation $g$ and correlated seeds $s_1,\ldots,s_\ell$ generated by   $\PRG_1$ and $\mathcal{H}'_g$ to be the hash family with this fixed allocation $g$ and independent seeds like \Cref{lem:independent_seeds}. 
    
    For any $\theta \in [M]$, the error between $\mathcal{H}_g$ and $\mathcal{H}'_g$ is at most
    \begin{align}
        & \left| \Pr_{h \sim \mathcal{H}_g}[h(y)=\theta \wedge \min h(X \setminus y)>\theta] - \Pr_{h'\sim \mathcal{H}'_g}[h'(y)=\theta \wedge \min h'(X \setminus y)>\theta] \right| \notag \\
        \le & \Pr_{\sigma \sim \mathcal{G}}[ \sigma(y)=\theta \wedge \min \sigma(B_j) > \theta] \cdot 2^{-C_s \cdot t}. \label{eq:after_NZ_PRG2} 
    \end{align}
\end{claim}

The term in \eqref{eq:after_NZ_PRG2} is the error introduced by $\PRG_1$, which is multiplicative with respect to $\Pr[\sigma(y)=\theta]$. So the last piece is to show the total error of \eqref{eq:after_NZ_PRG2} over $\theta$, $\sum_{\theta \in [M]} \Pr_{\sigma \sim \mathcal{G}}[ \sigma(y)=\theta \wedge \min \sigma(B_j) > \theta] \cdot 2^{-C_s \cdot t}$, is bounded by $o(\delta) / |X|$. The observation here is that $\sum_{\theta \in [M]} \Pr_{\sigma \sim \mathcal{G}}[ \sigma(y)=\theta \wedge \min \sigma(B_j) > \theta]$ is the exact probability of event $(\sigma(y) < \min \sigma(B_j))$ under $(0.1 C_s + 1)$-wise independence since $\mathcal{G}$ is $(0.1 C_s + 1)$-wise independent. Thus by the classical result of Indyk \cite{minwise_Indyk}, this part is bounded by $O(1)/|B_j|$.

We finish the proof of \Cref{thm:min_wise_log_seed} in \Cref{sec:proof_minwise}. Then we show the proofs of \Cref{lem:allocation_small_error}, \Cref{lem:independent_seeds} and \Cref{clm:dependent_seeds} in \Cref{sec:proof_allocation}, \Cref{sec:proof_independent_seeds} and \Cref{sec:proof_dependent_seeds} separately.

\subsection{Proof of \Cref{thm:min_wise_log_seed}}\label{sec:proof_minwise}

We continue the calculation of \eqref{eq:minwise_over_theta}: 
\begin{align*}
    \sum_{\theta \in [M]} \Pr_{w \sim U, (s_1, \ldots, s_\ell) \sim \PRG_1} \left[ (\sigma_j(y) = \theta \wedge \min \sigma_j(B_j) > \theta) \wedge (\min \sigma_i(B_i) > \theta,\ \forall i \ne j)  \right].    
\end{align*}

We first fix the allocation $g$. Then by \Cref{clm:dependent_seeds}, the above is equal to
\begin{equation*}
    \sum_{\theta \in [M]} \Pr_{h'\sim \mathcal{H}'_g}[h'(y)=\theta \wedge \min h'(X \setminus y)>\theta] \pm \sum_{\theta \in [M]} \Pr_{\sigma \sim \mathcal{G}}[ \sigma(y)=\theta \wedge \min \sigma(B_j) > \theta] \cdot 2^{-C_s \cdot t}.
\end{equation*}

\Cref{lem:independent_seeds} shows that the sum of the first term over allocations $g$ is $( 1 \pm 2^{-2 C \cdot t} ) \cdot \Pr_{\sigma \sim U}[\sigma(y) < \min \sigma(X \setminus y)]$. Then observe that given $g$, the second term becomes
\begin{equation}
    2^{-C_s \cdot t} \cdot \Pr_{\sigma : \text{$(0.1 C_s + 1)$-wise}}[\sigma(y) < \min \sigma(B_j)] = 2^{-C_s \cdot t} \cdot \frac{O(1)}{|B_j| + 1} \label{eq:error_C_s_wise},
\end{equation}

\noindent by Indyk's result that $O(\log (1 / \varepsilon))$-wise independence is a min-wise hash family with error $\varepsilon$ (choosing $\varepsilon$ as a constant here). For completeness, we provide a formal statement here.

\begin{theorem}[Theorem 1.1 in \cite{minwise_Indyk}]
    There exists a constant $c$ such that for any $\varepsilon > 0$, $c \cdot \log (1 / \varepsilon)$-wise independent function from $[N]$ to $[M]$ is a min-wise hash family with error $\varepsilon$ when $M = \Omega(N /\varepsilon)$. 
\end{theorem}

Returning to the second term, if $|X| < 2^{0.5 C_s \cdot t}$, then $2^{-C_s \cdot t} < 2^{-0.5 C_s \cdot t} / |X|$. Otherwise, $|X| \ge 2^{0.5 C_s \cdot t} > \ell^{1.1}$. By \Cref{lem:allocation_small_error}, with probability $1 - |X|^{-3 C}$, we have $|B_j| = (1 \pm 0.1) \cdot |X|/\ell$ in this case. And \eqref{eq:error_C_s_wise} becomes $\le 2^{-C_s \cdot t} \cdot \frac{O(1)}{|X| / \ell} = 2^{-(C_s - 1) \cdot t} \cdot \frac{O(1)}{|X|}.$ 

Finally, summing over all allocations $g$, we have
\begin{align*}
    \Pr_{h \sim \mathcal{H}}[h(y)< \min h(X \setminus y)] & = \Pr_{\sigma \sim U}[\sigma(y) < \min \sigma(X \setminus y)] \cdot (1 \pm 2^{-2 C \cdot t}) \pm \frac{O(1) \cdot 2^{-(C_s - 1) \cdot t} }{|X|} \pm 1 / |X|^{3 C} \\
    & = \Pr_{\sigma \sim U}[\sigma(y) < \min \sigma(X \setminus y)] \cdot (1 \pm 2^{-C \cdot t}),
\end{align*}

\noindent as $C_s$ is a large constant compared to $C$.

\subsection{Proof of \Cref{lem:allocation_small_error}}\label{sec:proof_allocation}

For convenience, set $r := |X| - 1$ in this proof. We prove these three cases separately.
\begin{enumerate}
    \item For $|X| \le \ell^{0.9}$, given $C_g \gg C$, we have
    $$
    \Pr_{g : \text{$C_g$-wise}}[|B_i| \ge C_g] \le {r \choose C_g} \cdot (1 / \ell)^{C_g} \le (r / \ell)^{C_g} \le 1 / \ell^{0.1 C_g} \le 1 / \ell^{3 C + 1},\ \forall i \in [\ell].
    $$

    \noindent By a union bound, with probability $1 - 1 / \ell^{3 C}$, $|B_i| \le C_g$ for all buckets.

    \item For $|X| \in (\ell^{0.9}, \ell^{1.1})$, let us fix a bin $i \in [\ell]$ and define $Z_x := \mathbf{1}(g(x) = i)$. Thus $|B_i| = \sum_{x \in X \setminus y} Z_x$ and $\E_{g : \text{$C_g$-wise}}\left[ (|B_i| - \E[|B_i|])^{C_g} \right] \le O( C_g \cdot r / \ell)^{C_g / 2}$, and 
    $$
    \Pr_{g : \text{$C_g$-wise}}[|B_i| - \E[|B_i|] \ge \ell^{0.1}] \le \frac{O(C_g \cdot r / \ell)^{C_g / 2}}{\ell^{0.1 C_g}} \le \frac{O(C_g \cdot \ell^{0.1})^{C_g / 2}}{\ell^{0.1 C_g}} = \dfrac{O_{C_g}(1)}{\ell^{0.05 C_g}} \le 1/\ell^{3 C + 1}.
    $$

    \noindent After a union bound, with probability $1 - 1 / \ell^{3 C}$, $|B_i| \le \ell^{0.1} + \E[|B_i|] \le 2 \ell^{0,1}$ for all $i$.

    \item For $|X| \ge \ell^{1.1}$, similar to the above analysis, we fix $i \in [\ell]$ and define $Z_x := \mathbf{1}(g(x) = i)$. However, the deviation depends more on $r$:
    $$
    \Pr_{g : \text{$C_g$-wise}}[||B_i| - \E[|B_i|]| \ge 0.1 \cdot r / \ell] \le \frac{O(C_g \cdot r / \ell)^{C_g / 2}}{(0.1 \cdot r / \ell)^{C_g}} \le \frac{O_{C_g}(1)}{(r / \ell)^{C_g / 2}} \le \frac{O_{C_g}(1)}{(r^{0.05})^{C_g / 2}} \le 1/r^{3 C + 1}.
    $$

    \noindent Using a union bound, with probability $1 - r^{-3 C}$, $|B_i| = (1 \pm 0.1) \cdot r / \ell$ for every $i \in [\ell]$.
\end{enumerate}

\subsection{Proof of \Cref{lem:independent_seeds}}\label{sec:proof_independent_seeds}

This proof relies on the extra property of our extractors in \Cref{lem:extractor_unif} and the fact that permuting the input bits will not affect the Nisan-Zuckerman PRG.

This proof has two parts. In the first part, given $X \subseteq [N]$, $y \in X$, $\theta \in [M]$ and the allocation mapping $g$, let $j:=g(y)$. We will prove 
\begin{align}            
    \Pr_{h' \sim \mathcal{H}'}\left[ h'(y) = \theta \wedge \min  h'(X \setminus y)>\theta \right] &
     = \frac{1}{M} \Pr_{\sigma : \text{$0.1 C_s$-wise}}\left[ \min \sigma(B_j) > \theta \right] \notag \\
     & \ \ \ \cdot \prod_{i \ne j} \left( (1 - \theta / M)^{|B_i|} \pm 2 \cdot 2^{-C_s \cdot t} \right) \pm \ell \cdot N^{-0.4 C_e}.  \label{eq:decompose_independent_seeds}
\end{align}

In the second part, we bound the summation of \eqref{eq:decompose_independent_seeds} over all $\theta \in [M]$ and allocations $g$ as $( 1 \pm 2^{-2 C \cdot t} ) \cdot \Pr_{\sigma \sim U}[\sigma(y) < \min \sigma(X \setminus y)]$. Note that \Cref{lem:allocation_small_error} tells we can ignore bad allocations.

\subsubsection{The First Part}\label{sec:first_proof_independent_seed}

We show \eqref{eq:decompose_independent_seeds} via the Nisan-Zuckerman framework. One subtle thing is that the extractor error $\NZerr = 2^{-C_s \cdot t}$ in \eqref{eq:decompose_independent_seeds} is multiplicative with respect to $\Pr_{\sigma : \text{$(0.1 C_s + 1)$-wise}}[\sigma(y) = \theta] = 1 / M$, different from applying the PRG fooling combinatorial rectangles directly. This is shown by permuting the order of $s_1,\ldots,s_\ell$ such that $z_j:=\Ext(w,s_j)$ is the first input of the read-once branching program.

Given the allocation $g$, we consider a width-$2$ length-$\ell$ read-once branching program $P$ whose input alphabet is $\{0, 1\}^{0.3 C_e \cdot \log N}$ corresponding to $z_1,\ldots,z_\ell$. Because of the definition $B_i:=\{x \in X \setminus y:g(x)=i\}$ and $h'(i):=\sigma_{g(i)}(i)$ for $\sigma_1:=G_{z_1},\ldots,\sigma_\ell=G_{z_\ell}$, we fix $g$ and rewrite $(h'(y)=\theta \wedge \min h'(X \setminus y)>\theta)$ as the conjunction of $\ell$ events depending on $z_1,\ldots,z_\ell$:
$$
\underbrace{\left( \sigma_{j}(y) = \theta \wedge \min \sigma_{j}(B_j) > \theta \right)}_{A_1} \wedge \cdots \wedge \underbrace{\left( \min \sigma_{\ell}(B_\ell) > \theta \right)}_{A_{\ell - j + 1}} \wedge \cdots \wedge \underbrace{\left( \min \sigma_{j - 1}(B_{j - 1}) > \theta \right)}_{A_\ell}.
$$
First of all, choosing this order will guarantee that $\sigma_{j}(y)=\theta$ is in the first event $A_1$. Secondly, this is a ROBP of width-2 with input $z_1,\ldots,z_\ell$.

Since $\sigma_i = G_{z_i}$ for $z_i = \Ext(w, s_i)$,
\begin{align}   
& \Pr_{w \sim U, (s_1, \ldots, s_\ell) \sim U : z_i = \Ext(w, s_i)}\left[ (\sigma_j(y) = \theta \wedge \min \sigma_j(B_j) > \theta) \wedge (\min \sigma_i(B_i) > \theta ,\ \forall i \ne j) \right] \notag \\
= & \Pr_{w \sim U, (s_1, \ldots, s_\ell) \sim U : z_i = \Ext(w, s_i)}\left[ A_1 \wedge A_2 \wedge \cdots \wedge A_\ell \right] \label{eq:ROBP} 
\end{align}
has correlated functions $\sigma_1, \ldots, \sigma_\ell$ generated from the Nisan-Zuckerman PRG with an extractor error $\NZerr$. However, because the \emph{first} input $z_j = \Ext(w, s_j)$ is uniform by \Cref{lem:extractor_unif}, the first function $\sigma_j = G_{z_j}$ guarantees that the first event happens as same as a uniform sampling in $\mathcal{G}$:
$$
\Pr_{w, s_j}[A_1] = \Pr_{w,s_j}[ \sigma_j(y) = \theta \wedge \min \sigma_j(B_j) > \theta ] = \Pr_{\sigma \sim \mathcal{G}}[ \sigma(y) = \theta \wedge \min \sigma(B_j) > \theta].
$$

Then we consider the rest events. Similar to the analysis of the Nisan-Zuckerman PRG, there are two cases for each $i = 2, 3, \ldots, \ell$: 
\begin{enumerate}
    \item When $\Pr_{w, s_j \ldots}[A_1 \wedge \cdots \wedge A_{i - 1}] \ge 2^{-0.4 C_e \cdot \log N}$, consider  the distribution of $w$ conditioned upon $A_1 \wedge \cdots \wedge A_{i - 1}$. Note that the conditioning only increases the probability of each value of $w$ by a factor of at most $2^{0.4 C_e \cdot \log N}$. Hence this conditional distribution has min-entropy at least $0.6 C_e \cdot \log N$. Then by the property of the extractor, we have
    $$
    \Pr_{w, s_j \ldots}[A_i \mid A_1 \wedge \cdots \wedge A_{i - 1}] = \Pr_{z_{j + i - 1} \sim U}[A_i] \pm \NZerr.
    $$

    Hence
    \begin{align*}
        \Pr_{w, s_j \ldots}[A_1 \wedge \cdots \wedge A_i] & = \Pr_{w, s_j \ldots}[A_1 \wedge \cdots \wedge A_{i - 1}] \cdot \Pr_{w, s_j \ldots}[A_i \mid A_1 \wedge \cdots \wedge A_{i - 1}] \\
         & = \Pr_{w, s_j \ldots}[A_1 \wedge \cdots \wedge A_{i - 1}] \cdot \left( \Pr_{z_{j + i - 1} \sim U}[A_i] \pm \NZerr \right).
    \end{align*}

    \item Otherwise, $\Pr_{w, s_j \ldots}[A_1 \wedge \cdots \wedge A_{i - 1}] < 2^{-0.4 C_e \cdot \log N}$ indicates
    $\Pr_{w, s_j \ldots}[A_1 \wedge \cdots \wedge A_i] \le 2^{-0.4 C_e \cdot \log N}$.
\end{enumerate}

Combining two cases together, we can write the the acceptance of the first $i$ events as
$$
\Pr_{w,s_j,\ldots}[A_1 \wedge \cdots \wedge A_i] = \Pr_{w,s_j,\ldots}[A_1 \wedge \cdots \wedge A_{i - 1}] \cdot \left( \Pr_{z_{j + i - 1} \sim U}[A_i] \pm \NZerr \right) \pm 2^{-0.4 C_e \cdot \log N}.
$$

Then by induction on $i$, the probability in \eqref{eq:ROBP} is expressed as
\begin{align*}
      & \Pr_{z_j \sim U}[A_1] \cdot \left( \Pr_{z_{j + 1} \sim U}[A_2] \pm \NZerr \right) \cdots \left( \Pr_{z_{j - 1} \sim U}[A_\ell] \pm \NZerr \right) \pm \ell \cdot 2^{-0.4 C_e \cdot \log N} \\
    = & \Pr_{\sigma \sim \mathcal{G}}\left[ \sigma(y) = \theta \wedge \min \sigma(B_j) > \theta \right] \cdot \prod_{i \ne j} \left( \Pr_{\sigma \sim \mathcal{G}}\left[ \min \sigma(B_i) > \theta \right] \pm \NZerr \right) \pm \ell \cdot N^{-0.4 C_e}. 
\end{align*}

Since $\mathcal{G}$ is $(0.1 C_s + 1)$-wise independent, for $A_1$, it follows that
$$
\Pr_{\sigma \sim \mathcal{G}}[\sigma(y) = \theta \wedge \min \sigma(B_j) > \theta] = \frac{1}{M} \Pr_{\sigma : 0.1 C_s\text{-wise}}[\min \sigma(B_j) > \theta].
$$

$\mathcal{G}$ is also the PRG for combinatorial rectangles with error $2^{-C_s \cdot t}$, so for the remaining events, we simplify their probabilities as
$$
\Pr_{\sigma \sim \mathcal{G}}[\min \sigma(B_i) > \theta] \pm \NZerr = \Pr_{\sigma \sim U}[\min \sigma(B_i) > \theta] \pm 2^{-C_s \cdot t} \pm \NZerr = (1 - \theta / M)^{|B_i|} \pm 2 \cdot 2^{-C_s \cdot t}.
$$

This shows \eqref{eq:decompose_independent_seeds}:
\begin{align*}
    \Pr_{h' \sim \mathcal{H}'}\left[ h'(y) = \theta \wedge \min  h'(X \setminus y)>\theta \right] & =
    \dfrac{1}{M} \Pr_{\sigma : \text{$0.1 C_s$-wise}}\left[ \min \sigma(B_j) > \theta \right] \notag \\
     & \quad \cdot \prod_{i \ne j} \left( (1 - \theta / M)^{|B_i|} \pm 2 \cdot 2^{-C_s \cdot t} \right) \pm \ell \cdot N^{-0.4 C_e}.
\end{align*}

\subsubsection{The Second Part}\label{sec:second_proof_independent_seeds}

Here we show the summation of \eqref{eq:decompose_independent_seeds} over all $\theta$'s and allocations $g$ equals $(1 \pm 2^{-2 C \cdot t}) \cdot \Pr_{\sigma \sim U}[\sigma(y) < \min \sigma(X \setminus y)]$. This indicates $\Pr_{h' \sim \mathcal{H}'}[h'(y) < \min h'(X \setminus y)]=(1 \pm 2^{-2 C \cdot t}) \cdot \Pr_{\sigma \sim U}[\sigma(y) < \min \sigma(X \setminus y)]$ by the first part of this proof and finishes the proof of \Cref{lem:independent_seeds} --- $\mathcal{H'}$ has a multiplicative error $2^{-2 C \cdot t}$.

For convenience, set $\varepsilon := 2 \cdot 2^{-C_s \cdot t}$ in this proof. Our rough plan is to apply the first approximation of \Cref{lem:t_wise_ind} for small $\theta$ to show that
$$
\Pr_{\sigma : 0.1 C_s \text{-wise}}[ \min \sigma(B_j) > \theta ] \cdot \prod_{i \ne j} \left((1 - \theta / M)^{|B_i|} \pm \varepsilon \right) \approx (1 - \theta / M)^{|X| - 1}.
$$

For large $\theta$ such that $(1-\theta/M)^{|X|-1}$ is tiny, we choose a suitable subset $T \subseteq [\ell] \setminus j$, and use the following fact to bound the tail probability:
$$
\Pr_{\sigma : 0.1 C_s \text{-wise}}[ \min \sigma(B_j) > \theta ] \cdot \prod_{i \ne j} \left((1 - \theta / M)^{|B_i|} \pm \varepsilon \right) \le \prod_{i \in T} \left( (1 - \theta / M)^{|B_i|} + \varepsilon \right).
$$

The actual calculation depends on the size of $X$. According to \Cref{lem:allocation_small_error}, we will split the calculations into three cases: (1) $|X| \le \ell^{0.9}$; (2) $|X| \in (\ell^{0.9}, \ell^{1.1})$; (3) $|X| \ge \ell^{1.1}$.

Moreover, since we have assumed $M = (N / \delta)^{O(1)}$ such that $\Pr_{\sigma \sim U}[\sigma(y) < \min \sigma(X \setminus y)] \approx 1 / |X|$, we treat multiplicative errors as additive errors multiplied by a factor of $|X|$ in this analysis.

\paragraph{The first case of $|X| \le \ell^{0.9}$.} We assume that each bucket has at most $C_g$ elements in $X \setminus y$ according to the first property of \Cref{lem:allocation_small_error}. The corresponding failure probability will change the final multiplicative error by at most $|X| / \ell^{3 C} \le 1 / \ell^{3 C - 0.9} < 2^{-2.5 C \cdot t}$.

As $C_s \gg C_g$, we assume $0.1 C_s > C_g \ge |B_j|$. Thus $\Pr_{\sigma : \text{$0.1 C_s$-wise}}[ \min \sigma(B_j) > \theta] = (1 - \theta / M)^{|B_j|}$ and
$$
\Pr_{\sigma : \text{$0.1 C_s$-wise}}[\min \sigma(B_j) > \theta] \cdot \prod_{i \ne j} \left( (1 - \theta / M)^{|B_i|} \pm \varepsilon \right) = \prod_{i \in [\ell]} \left( (1 - \theta / M)^{|B_i|} \pm \varepsilon \right).
$$

When $(1 - \theta / M)^{C_g} > 2^{-0.5 C_s \cdot t}$, we simplify the above additive error $\varepsilon = 2 \cdot 2^{-C_s \cdot t}$ as
$$
(1 - \theta / M)^{|B_i|} \pm \varepsilon = (1 - \theta / M)^{|B_i|} \cdot (1 \pm 2 \cdot 2^{-0.5 C_s\cdot t} ),\ \forall i \in [\ell].
$$

\noindent Hence
$$
\sum_{(1 - \theta / M)^{C_g} > 2^{-0.5 C_s \cdot t}} \prod_{i \in [\ell]} \left( (1 - \theta / M)^{|B_i|} \pm \varepsilon \right) =(1 \pm 4 \cdot 2^{-(0.5 C_s - 1) \cdot t}) \cdot \sum_{(1 - \theta / M)^{C_g} > 2^{-0.5 C_s \cdot t}} (1 - \theta / M)^{|X| - 1}.
$$

Otherwise, $\theta$ is large enough such that $(1 - \theta / M)^{C_g} \le 2^{-0.5 C_s \cdot t}$. Note that each term in the summation does not exceed $1$, and the number of such $\theta$'s is at most $M \cdot 2^{-0.5 C_s \cdot t / C_g}$. Thus we simply bound the sum over such $\theta$'s as
\begin{align*}
    \frac{1}{M} (1 - \theta / M)^{|X| - 1} = \frac{1}{M} \sum_{(1 - \theta / M)^{C_g} \le 2^{-0.5 C_s \cdot t}} \prod_{i \in [\ell]} \left( (1 - \theta / M)^{|B_i|} \pm \varepsilon \right) \le 2^{-0.5 C_s \cdot t / C_g}.
\end{align*}

Now we have 
\begin{align*}
      & \Pr_{h' \sim \mathcal{H}'}\left[ h'(y) = \theta \wedge \min  h'(X \setminus y)>\theta \right] \\
    = &
    \frac{1}{M} \left( \sum_{(1 - \theta / M)^{C_g} > 2^{-0.5 C_s \cdot t}} \prod_{i \in [\ell]} \left( (1 - \theta / M)^{|B_i|} \pm \varepsilon \right) + \sum_{(1 - \theta / M)^{C_g} \le 2^{-0.5 C_s \cdot t}} \prod_{i \in [\ell]} \left( (1 - \theta / M)^{|B_i|} \pm \varepsilon \right) \right) \pm 1/\ell^{3C} \\
    = &  \frac{1 \pm 4 \cdot 2^{-(0.5 C_s - 1) \cdot t}}{M} \cdot \sum_{(1 - \theta / M)^{C_g} > 2^{-0.5 C_s \cdot t}} (1 - \theta / M)^{|X| - 1} \pm 2 \cdot 2^{-0.5 C_s \cdot t / C_g} \pm 1/\ell^{3C}.
\end{align*}
Compared to the fair probability $\sum_{\theta \in [M]} \frac{1}{M} \cdot (1 - \theta / M)^{|X| - 1}$, the multiplicative error (w.r.t. $1/|X|$) of \eqref{eq:decompose_independent_seeds} is at most
$$
4 \cdot 2^{-(0.5 C_s - 1) \cdot t} +  \dfrac{1}{\ell^{3 C - 0.9}} + |X| \cdot 2^{-0.5 C_s \cdot t / C_g} \le 2^{-(0.5 C_s - 1) \cdot t + 2} + 2^{-2.5 C \cdot t} + 2^{-(0.5 C_s / C_g - 1) \cdot t} < 2^{-2 C \cdot t}.
$$
In the last step, We choose $C_s \gg C_g$ such that $0.5 C_s / C_g > 3 C$.

\paragraph{The second case of $|X| \in (\ell^{0.9}, \ell^{1.1})$.} Similarly, we assume $\max_{i \in [\ell]} |B_i|$ is at most $2 \ell^{0.1}$ by \Cref{lem:allocation_small_error}. The failure probability only affects the multiplicative error by at most $|X| / \ell^{3 C} \le 1 / \ell^{3 C - 1.1 } < 2^{-2.5 C \cdot t}$. 

When $\theta \le 0.5 M \cdot \ell^{-0.2}$, there comes $(1 - \theta / M)^{|B_i|} \ge 1 - |B_i| \cdot \theta / M \ge 1 - \ell^{-0.1} \ge 0.5,\ \forall i \in [\ell]$. We apply the first statement in \Cref{lem:t_wise_ind} to estimate $\Pr_{\sigma : \text{$0.1 C_s$-wise}}[ \min \sigma(B_j) > \theta ]$:
\begin{align*}
    \Pr_{\sigma : \text{$0.1 C_s$-wise}}[\min \sigma(B_j) > \theta] & = (1 - \theta / M)^{|B_j|} \pm \left( |B_j| \cdot \dfrac{\theta}{M} \right)^{0.1 C_s} \\
     & = (1 - \theta / M)^{|B_j|} \cdot \left( 1 \pm \dfrac{2}{\ell^{0.01 C_s}} \right).
\end{align*}

\noindent For the remaining buckets, we have
$$
(1 - \theta / M)^{|B_i|} \pm \varepsilon = (1 - \theta / M)^{|B_i|} \cdot \left( 1 \pm 2 \varepsilon \right),\ \forall i \ne j.
$$

\noindent Therefore, for small $\theta$ with $\theta / M \le 0.5 \ell^{-0.2}$, it holds that
\begin{align*}
     & \dfrac{1}{M} \sum_{\theta \le 0.5 M \cdot \ell^{-0.2}} \Pr_{\sigma : \text{$0.1 C_s$-wise}}[\min \sigma(B_j) > \theta] \cdot \prod_{i \ne j} \left( (1 - \theta / M)^{|B_i|} \pm \varepsilon \right) \\
    = & \dfrac{1}{M} \sum_{\theta \le 0.5 M \cdot \ell^{-0.2}} (1 - \theta / M)^{|X| - 1} \cdot \left( 1 \pm \left( \dfrac{4}{\ell^{0.01 C_s}} + 4 (\ell - 1) \varepsilon \right) \right) \\
    = & \dfrac{1}{M} \sum_{\theta \le 0.5 M \cdot \ell^{-0.2}} (1 - \theta / M)^{|X| - 1} \cdot \left( 1 \pm 2^{-0.001 C_s \cdot t} \right).
\end{align*}

Otherwise, when $\theta > 0.5 M \cdot \ell^{-0.2}$, we show the tail summations are small in both cases of $\sigma \sim U$ and $h' \sim \mathcal{H}'$, leaving a negligible additive error.

Note that both $\frac{1}{M} \sum_{\theta > 0.5 M \cdot \ell^{-0.2}} (1 - \theta / M)^{|X| - 1}$ and $\frac{1}{M} \sum_{\theta > 0.5 M \cdot \ell^{-0.2}} \Pr_{\sigma : \text{$0.1 C_s$-wise}}[\min \sigma(B_j) > \theta] \cdot \prod_{i \ne j} \left( (1 - \theta / M)^{|B_i|} \pm \varepsilon \right)$ are upper bounded by $\frac{1}{M} \sum_{\theta > 0.5 M \cdot \ell^{-0.2}} \prod_{i \ne j} \left( (1 - \theta / M)^{|B_i|} + \varepsilon \right)$. Hence, we focus on the latter one, $\frac{1}{M} \sum_{\theta > 0.5 M \cdot \ell^{-0.2}} \prod_{i \ne j} \left( (1 - \theta / M)^{|B_i|} + \varepsilon \right)$.

Consider the sizes of all buckets, say $|B_1|, \ldots, |B_\ell|$, and define $b$ to be the $S:=C' \cdot \log \log N$ largest number among $|B_1|, \ldots, |B_\ell|$ excluding $|B_j|$. Without loss of generality, assume $j = \ell$ and $2 \ell^{0.1} \ge |B_1 | \ge |B_2 | \ge \cdots \ge |B_{\ell - 1}|$, then $b = |B_S|$. 

We split the sum into two cases depending on whether $(1 - \theta / M)^b > \varepsilon \cdot \ell$.
\begin{enumerate}
    \item For $(1 - \theta / M)^b > \varepsilon \cdot \ell$, we further simplify it as
    \begin{align*}
         & \dfrac{1}{M} \sum_{\theta > 0.5 M \cdot \ell^{-0.2} \wedge (1 - \theta / M)^b > \varepsilon \cdot \ell} \prod_{i \ne j} \left( (1 - \theta / M)^{|B_i|} + \varepsilon \right) \\
        \le & \dfrac{1}{M} \sum_{\theta > 0.5 M \cdot \ell^{-0.2} \wedge (1 - \theta / M)^b > \varepsilon \cdot \ell} \prod_{i = S + 1}^{\ell - 1} (1 - \theta / M)^{|B_i|} \cdot (1 + 1 / \ell) \\
        \le & \dfrac{1}{M} \sum_{\theta > 0.5 M \cdot \ell^{-0.2} \wedge (1 - \theta / M)^b > \varepsilon \cdot \ell} e \cdot (1 - \theta / M)^{|X| - 1 - (S + 1) \cdot 2 \ell^{0.1}} \\
        \le & \dfrac{e}{M} \sum_{\theta > 0.5 M \cdot \ell^{-0.2} \wedge (1 - \theta / M)^b > \varepsilon \cdot \ell} (1-\theta/M)^{0.9 |X|}\\
        \le & \exp( 1 - 0.9 |X| \cdot 0.5 \ell^{-0.2} ) \le \exp(-\ell^{0.6}).
    \end{align*}

    \item For $(1 - \theta / M)^b \le \varepsilon \cdot \ell$, we bound it as
    \begin{align*}
         & \dfrac{1}{M} \sum_{\theta > 0.5 M \cdot \ell^{-0.2} \wedge (1 - \theta / M)^b \le \varepsilon \cdot \ell} \prod_{i \ne j} \left( (1 - \theta / M)^{|B_i|} + \varepsilon \right) \\
        \le & \dfrac{1}{M} \sum_{\theta > 0.5 M \cdot \ell^{-0.2} \wedge (1 - \theta / M)^b \le \varepsilon \cdot \ell} \prod_{i = 1}^S (\varepsilon \cdot \ell + \varepsilon) \\
        \le & \dfrac{1}{M} \sum_{\theta > 0.5 M \cdot \ell^{-0.2} \wedge (1 - \theta / M)^b \le \varepsilon \cdot \ell} \prod_{i = 1}^S 2^{2 - (C_s - 1) \cdot t} \\
        \le & 2^{2 S - (C_s - 1) \cdot t \cdot S} = N^{-O(1)},
    \end{align*}
    where we plug the definition of $S = C' \cdot \log \log N$ in the last step.
\end{enumerate}

So, the total multiplicative error is bounded by
$$
\dfrac{1}{\ell^{3 C - 1.1}} + 2^{-0.001 C_s \cdot t} + |X| \cdot \left( \exp(-\ell^{0.6}) + N^{-O(1)} \right) < 2^{-2 C \cdot t}.
$$

\paragraph{The third case of $|X| \ge \ell^{1.1}$.} The proof for this case is identical to the second case. Here, the max-load is bounded according to $1.1 \cdot |X| / \ell$ by \Cref{lem:allocation_small_error}. The threshold of applying the tail bound would be between $\frac{|X|}{\ell} \cdot \frac{\theta}{M} \le \ell^{-0.1}$ and  $\frac{|X|}{\ell} \cdot \frac{\theta}{M} > \ell^{-0.1}$, while the rest calculation is very similar and we leave it in \hyperref[sec:proof_X_big_ind_seed]{Appendix B}.

\subsection{Proof of \Cref{clm:dependent_seeds}}\label{sec:proof_dependent_seeds}

After fixing the allocation $g$, since we have proven $\Pr_{h'\sim \mathcal{H}'_g}[h'(y)=\theta \wedge \min h'(X \setminus y)>\theta]$ equals 
\begin{equation}\label{eq:prob_h'}
    \frac{1}{M} \Pr_{\sigma : \text{$0.1 C_s$-wise}}\left[ \sigma(y)=\theta \wedge \min \sigma(B_j) > \theta \right] \cdot \prod_{i \ne j} \left( (1 - \theta / M)^{|B_i|} \pm 2 \cdot 2^{-C_s \cdot t} \right) \pm \ell \cdot N^{-0.4 C_e}    
\end{equation}
in \Cref{sec:first_proof_independent_seed}  (as equation \eqref{eq:decompose_independent_seeds}), this proof will show the error between $\Pr_{h \sim \mathcal{H}_g}[h(y)=\theta \wedge \min h(X \setminus y)>\theta]$ and \eqref{eq:prob_h'} is at most $\Pr_{\sigma \sim \mathcal{G}}[ \sigma(y)=\theta \wedge \min \sigma(B_j) > \theta] \cdot 2^{-C_s \cdot t}$.

We need the following property of $s_1, \ldots, s_\ell$ when they are generated by $\PRG_1$.

\begin{fact}\label{fact:PRG_CR_LOG}
    Those random seeds $s_1,\ldots,s_\ell$ satisfy that  for any fixed $j \in [\ell]$, any string $\alpha \in \{0, 1\}^{C_e \cdot t}$, and any $\ell$ functions $f_1, \ldots, f_\ell : \{0, 1\}^{C_e \cdot t} \to \{0, 1\}$, it has
    $$
    \E_{(s_1,\ldots,s_\ell) \sim \PRG_1}\left[ \prod_{i \neq j}f_i(s_i)\ \middle|\ s_j=\alpha \right] = \prod_{i \neq j} \E_{s_i \sim U}[f_i(s_i)] \pm 2^{-C_s \cdot t}.
    $$
\end{fact}

\begin{proof}
    Note that both $\mathbf{1}(s_j = \alpha)$ and $\mathbf{1}(s_j = \alpha) \cdot \prod_{i \ne j} f_i(s_i)$ are combinatorial rectangles in $\left( \{0, 1\}^{C_e \cdot t} \right)^\ell$. According to the property of $\PRG_1$, we have
    $$
    \Pr_{s_j \sim \PRG_1}[s_j = \alpha]=2^{-C_e \cdot t} \pm 2^{-(C_s + C_e) \cdot t - 2}, \text{\ and}
    $$
    $$
    \E_{(s_1, \ldots, s_\ell) \sim \PRG_1}\left[ \mathbf{1}(s_j = \alpha) \cdot \prod_{i \neq j} f_i(s_i) \right] = \E_{(s_1,\ldots,s_\ell) \sim U}\left[ \mathbf{1}(s_j = \alpha) \cdot \prod_{i \neq j} f_i(s_i) \right] \pm 2^{-(C_s + C_e) \cdot t - 2}.
    $$

    So we can express the conditional expectation as
    \begin{align*}
        \E_{(s_1,\ldots,s_\ell) \sim \PRG_1}\left[ \prod_{i \neq j}f_i(s_i)\ \middle|\ s_j=\alpha \right] & = \frac{\E_{(s_1,\ldots,s_\ell) \sim \PRG_1}\left[ \mathbf{1}(s_j = \alpha) \cdot \prod_{i \neq j}f_i(s_i) \right]}{\Pr_{s_j \sim \PRG_1}[s_j = \alpha]} \\
         & = \frac{\prod_{i \ne j} \E_{s_i \sim U}[f_i(s_i)] \pm 2^{-C_s \cdot t - 2}}{1 \pm 2^{-C_s \cdot t - 2}}.
    \end{align*}

    \noindent Note that $\prod_{i \ne j} \E_{s_i \sim U}[f_i(s_i)] \le 1$, which tells that the additive error is at most
    $$
    \frac{\prod_{i \ne j} \E_{s_i \sim U}[f_i(s_i)] + 2^{-C_s \cdot t - 2}}{1 - 2^{-C_s \cdot t - 2}} - \prod_{i \ne j} \E_{s_i \sim U}[f_i(s_i)] \le \left( \frac{1}{1 - 2^{-C_s \cdot t - 2}} - 1 \right) + \frac{2^{-C_s \cdot t - 2}}{1 - 2^{-C_s \cdot t - 2}} \le 2^{-C_s \cdot t},
    $$
    as desired.
\end{proof}

Now we are ready to finish the proof of \Cref{clm:dependent_seeds} in this section. We rewrite the probability $\Pr_{h \sim \mathcal{H}_g}[h(y) = \theta \wedge \min h(X  \setminus y) > \theta]$ as (recall $\sigma_i:=G_{z_i}$ for $z_i:=\Ext(w,s_i)$)
\begin{align}
      & \Pr_{w \sim U, (s_1, \ldots, s_\ell) \sim \PRG_1} \left[ (\sigma_j(y) = \theta \wedge \min \sigma_j(B_j) > \theta) \wedge (\min \sigma_i(B_i) > \theta,\ \forall i \ne j)  \right] \notag \\
    = & \sum_\alpha \E_{w \sim U}\left[ \Pr_{(s_1, \ldots, s_{\ell}) \sim \PRG_1}[(s_j = \alpha \wedge \sigma_j(y) = \theta \wedge \min \sigma_j(B_j) > \theta) \wedge (\min \sigma_i(B_i) > \theta,\ \forall i \ne j)] \right] \notag \\
    = & \sum_\alpha \E_{w \sim U} \left[ \Pr_{s_j \sim \PRG_1}[ s_j = \alpha \wedge \sigma_j(y) = \theta \wedge \min \sigma_j(B_j)>\theta ] \cdot \Pr_{(s_1, \ldots, s_\ell) \sim \PRG_1}[ \min \sigma_i(B_i) > \theta,\ \forall i \ne j \mid s_j = \alpha ] \right], \label{eq:Before_PRG_RC}
\end{align}

\noindent where the last line applies the fact that only those $s_j = \alpha$ which make $(\sigma_j(y) = \theta \wedge \min \sigma_j(B_j) > \theta)$ true have contributions to the expectation.

For a fixed $w$, the first event $(\sigma_j(y) = \theta \wedge \min \sigma_j(B_j) > \theta)$ is determined since $s_j=\alpha$ is fixed. Then the rest events $\mathbf{1}(\min \sigma_i(B_i) > \theta,\ \forall i \ne j)$ for $\sigma_i = G_{z_i}$ and $z_i = \Ext(w, s_i)$ constitute a combinatorial rectangle of $s_1,\ldots,s_\ell$ in $\left( \{0, 1\}^{C_e \cdot t} \right)^\ell$. Then by \Cref{fact:PRG_CR_LOG},
$$
\Pr_{(s_1, \ldots, s_\ell) \sim \PRG_1} [ \min \sigma_i(B_i) > \theta,\ \forall i \ne j \mid s_j = \alpha ] = \prod_{i \ne j}\Pr_{s_i \sim U}[\min \sigma_i(B_i) > \theta ] \pm 2^{-C_s \cdot t}.
$$

So we apply this to simplify \eqref{eq:Before_PRG_RC} as
\begin{align}
      & \sum_\alpha \E_{w \sim U} \left[ \Pr_{s_j \sim \PRG_1} [ s_j = \alpha \wedge \sigma_j(y) = \theta \wedge \min \sigma_j(B_j) > \theta  ] \cdot \left( \prod_{i \ne j}\Pr_{s_i \sim U}[\min \sigma_i(B_i) > \theta] \pm 2^{-C_s \cdot t} \right) \right] \notag \\
    = & \sum_\alpha \Pr_{w \sim U, s_j \sim \PRG_1}[s_j = \alpha \wedge \sigma_j(y) = \theta \wedge \min \sigma_j(B_j) > \theta] \notag \\
      & \qquad \cdot \left( \Pr_{w \sim U, (s_i)_{i \ne j} \sim U}[\min \sigma_i(B_i) > \theta,\ \forall i \ne j \mid \sigma_j(y) = \theta \wedge \min \sigma_j(B_j) > \theta] \pm 2^{-C_s \cdot t} \right). \label{eq:before_NZ_PRG}
\end{align}


We use properties of our extractor to simplify \eqref{eq:before_NZ_PRG}. Since $\Ext(w, \alpha) = U$ when $w$ is uniform and $\alpha$ is fixed, $\sigma_j$ in the first event $(\sigma_j(y) = \theta \wedge \min \sigma_j(B_j) > \theta)$ conditioned on $s_j = \alpha$ is uniformly sampled from $\mathcal{G}$. Thus this event holds with probability
\begin{align*}
     & \Pr_{w \sim U, s_j \sim \textsf{PRG}_1}[s_j = \alpha \wedge \sigma_j(y) = \theta \wedge \min \sigma_j(B_j) > \theta] \\
    = & \Pr_{s_j \sim \textsf{PRG}_1}[s_j = \alpha] \cdot \Pr_{w \sim U, s_j \sim \PRG_1}[\sigma_j(y) = \theta \wedge \min \sigma_j(B_j) > \theta \mid s_j = \alpha] \\
    = & \Pr_{s_j \sim \textsf{PRG}_1}[s_j = \alpha] \cdot \Pr_{\sigma \sim \mathcal{G}}[\sigma(y) = \theta \wedge \min \sigma(B_j) > \theta].
\end{align*}

Next we consider the term $\Pr_{w \sim U, (s_i)_{i \ne j} \sim U}[\min \sigma_i(B_i) > \theta,\ \forall i \ne j \mid \sigma_j(y) = \theta \wedge \min \sigma_j(B_j) > \theta]$ in \eqref{eq:before_NZ_PRG}. Following the same analysis in \Cref{sec:first_proof_independent_seed}, it equals
\begin{align*}
    \prod_{i \ne j} \left( (1 - \theta / M)^{|B_i|} \pm 2 \cdot 2^{-C_s \cdot t} \right) \pm \ell \cdot N^{-0.4 C_e}.
\end{align*}

Combining all equations to simplify \eqref{eq:before_NZ_PRG}, we finish the proof:
\begin{align*}
      & \Pr_{h \sim \mathcal{H}_g}[h(y) = \theta \wedge \min h(X \setminus y) > \theta] \\
    = & \sum_\alpha \Pr_{s_j \sim \PRG_1}[s_j = \alpha] \cdot \Pr_{\sigma \sim \mathcal{G}}[ \sigma(y) = \theta \wedge \min \sigma(B_j) > \theta ] \cdot \left( \prod_{i \ne j} \left( (1 - \theta / M)^{|B_i|} \pm 2 \cdot 2^{-C_s \cdot t} \right) \pm 2^{-C_s \cdot t}\right) \\
    = & \frac{1}{M} \Pr_{\sigma : \text{$0.1 C_s$-wise}}[\min \sigma(B_j) > \theta] \cdot \prod_{i \ne j} \left( (1-\theta/M)^{|B_i|} \pm 2 \cdot 2^{-C_s \cdot t} \right) \pm \Pr_{\sigma \sim \mathcal{G}}[ \sigma(y) = \theta \wedge \min \sigma(B_j) > \theta ] \cdot 2^{-C_s \cdot t},
\end{align*}

\noindent where the first term in the last line matches \eqref{eq:prob_h'}.
\section{$k$-min-wise Hash}\label{sec:k_min_wise}

We use the second approach outlined in \Cref{sec:overview} to construct $k$-min-wise hash in this section. Recall the definition that $\mathcal{H} = \{h_i : [N] \to [M]\}$ is a $k$-min-wise hash family of error $\delta$, if for any $X \subseteq [N]$ and $Y \subseteq X$ of size at most $k$, $\Pr_{h \sim \mathcal{H}}[\max h(Y) < \min h(X \setminus Y)] = \frac{1 \pm \delta}{{|X| \choose |Y|}}$. Without loss of generality, we assume $|Y|=k$ in this section.

\begin{theorem}\label{thm:k_min_wise}
    Given any $N$, $k = \log^{O(1)} N$, and any constant $C$, there exists an explicit $k$-min-wise hash family such that its seed length is $O_C(k \log N)$ and its multiplicative error is $\delta = 2^{-C \cdot \frac{\log N}{\log \log N}}$.
\end{theorem}

While our construction is in a similar framework of the construction in \Cref{sec:smaller_error}, there are several differences. The first one is to direct sum the output with an $O(k)$-wise independent function $h_0$ in the last step. The second difference is in the analysis, $Y$ could be in different buckets such that the first approach of guaranteeing its $\Ext(w,s_j)=U$ for $j:=g(y)$ does not work anymore. Instead of it, we will consider the conditional event $\Pr[\min h(X \setminus Y)>\theta \mid \max h(Y)=\theta]$ in this proof.

\begin{tcolorbox}
    \textbf{$k$-min-wise Hash Family: dimension $N$, $k = \log^{O(1)} N$, error $2^{-C \cdot \frac{\log N}{\log \log N}}$, and alphabet $N^{O(1)}$}

    \begin{enumerate}
        \item Set $\ell := 2^{t}$ for $t := \frac{\log N}{\log \log N}$ and pick large constants $C_e$, $C_s$ and $C_g$ such that $C_e > C_s > C_g > C$.
    
        \item Sample a $C_g \cdot k$-wise independent function $g : [N] \to [\ell]$ as the allocation of $[N]$ into $\ell$ buckets.

        \item Apply $\PRG_1$ in \Cref{thm:PRG_comb_rect} fooling combinatorial rectangles of seed length $O(k \log N)$ to generate $\ell$ seeds $s_1,\ldots,s_{\ell}$ in $\{0,1\}^{C_e t}$ of error $2^{-(C_s+C_e) \cdot t \cdot k - 2}$.
        
        \item Sample a source $w \sim \{0,1\}^{10 k \cdot C_e \log N}$ and apply an extractor $\Ext:\{0,1\}^{10 k \cdot C_e \log N} \times \{0,1\}^{C_e t} \rightarrow \{0,1\}^{C_e \cdot \log N}$ from \Cref{lem:extractor_unif} of min-entropy $6 k C_e \log N$ and error $\NZerr=2^{- C_s \cdot  t}$ to $w$ and $s_1,\ldots,s_\ell$: let $z_i := \Ext(w, s_i), \forall i \in [\ell]$. 
        
        \item Let $\sigma_i := \PRG_2(z_i)$ where $\PRG_2$ fools combinatorial rectangles in $[M]^N$ with  error $2^{-C_s \cdot t}$ from \Cref{cor:almost_poly_small_error}.

		\item Define $\varphi(x) := \sigma_{g(x)}(x)$ for every $x$. 
		
        \item Choose $h_0:[N] \rightarrow [M]$ from a $(C_e +1) \cdot k$-wise independent hash family.
            
        \item  Output the direct sum $h := h_0 + \varphi$.
    \end{enumerate}
\end{tcolorbox}

One remark is that Step 3 restrains $k = \log^{O(1)} N$ in order to guarantee the seed length of $\PRG_1$ is $O(k \log N)$ bits. In the rest of this section, we prove \Cref{thm:k_min_wise} and finish its analysis.

Similar to \Cref{lem:allocation_small_error}, we have the following lemma about the allocation of $X$ under $g$. Let $B_i:=\{x \in X \setminus Y: g(x)=i\}$ be the elements in $X \setminus Y$ mapped to bucket $i$ and $J:=\{j_1,\ldots,j_{k'}\}$ be the buckets in $[\ell]$ that contains elements in $Y$ (under $g$), i.e., $J:=\{g(y): y \in Y\}$. And let $B_J := \bigcup_{i = 1}^{k'} B_{j_i}$. 

\begin{lemma}\label{lem:allocation_k_min_wise}
    Let $C_g$ be a sufficiently large constant compared to $C$. Then $g$ guarantees that:
    \begin{enumerate}
        \item When $|X| \le \ell^{0.9}$, with probability $1 - \frac{1}{\ell^{3 C} \cdot |X|^{k}}$, $|B_i| \le C_g + 10 \cdot \frac{k \log |X|}{\log N / \log \log N}$ for all $i \in [\ell]$ and $|B_J| \le C_g \cdot k$.

        \item When $|X| \in (\ell^{0.9}, \ell^{1.1})$, with probability $1 - 1 / \ell^{3 C \cdot k}$, the max-load $\max_{i \in [\ell]} |B_i| \le 2 \ell^{0.1}$.
        
        \item When $|X| \ge \ell^{1.1}$, with probability $1-|X|^{-3 C \cdot k}$, all buckets satisfy $|B_i| = (1 \pm 0.1) \cdot |X| / \ell$.
    \end{enumerate}
\end{lemma}

Similar to the analysis in \Cref{thm:min_wise_log_seed}, the failure probability in \Cref{lem:allocation_k_min_wise} is relatively small compared to $\delta/{|X| \choose |Y|} \approx \delta/|X|^k$. So we \emph{fix $g$} and assume all properties in \Cref{lem:allocation_k_min_wise} hold in this section. The proof of this lemma resembles \Cref{lem:allocation_small_error}. We defer its proof to \hyperref[sec:proof_allocation_k_min_wise]{Appendix C}.

However, we can not assume the allocation of $Y$ because we can not compare its failure probability with $1/|X|^k$. For example, the probability that a bucket has $\log \log N$ elements in $Y$ is at least $1/N$ since the number of buckets $\ell=2^{\frac{\log N}{\log \log N}}$. Since $1/|X|$ could be as small as $1/N$, this implies that even for $Y$ as small as $\log \log N$, we could not guarantee $Y$ is uniformly distributed over $\ell$ buckets. 

We rewrite $\Pr_{h \sim \mathcal{H}}[\max h(Y) < \min h(X \setminus Y)]$ by enumerating $\theta=\max h(Y)$:
\begin{equation}
    \Pr_{h \sim \mathcal{H}}[\max h(Y) < \min h(X \setminus Y)] = \sum_{\theta \in [M]}  \Pr_{h \sim \mathcal{H}}\left[ \max h(Y) = \theta \wedge \min h(X \setminus Y) > \theta \right] \label{eq:all_events}
\end{equation}

Similar to \Cref{clm:dependent_seeds} in the proof of \Cref{thm:min_wise_log_seed}, we  simplify \eqref{eq:all_events} to events with independent seeds.

\begin{claim}\label{clm:expand}
    Recall $B_J := \bigcup_{i = 1}^{k'} B_{j_i}$ contains buckets in $[\ell]$ with elements in $Y$ (under $g$). 
    
    $\Pr_{h \sim \mathcal{H}}\left[ \max h(Y) = \theta \wedge \min h(X \setminus Y) > \theta \right]$ in \eqref{eq:all_events} could be decomposed as
    \begin{align}
         \Pr_{\sigma : (C_e+1) \cdot k \text{-wise}}[\max \sigma(Y)=\theta \wedge \min \sigma(B_J) > \theta ] \cdot \left( \prod_{i \notin J} \left(\Pr_{\sigma \sim U}[\min \sigma(B_i) > \theta] \pm 2 \cdot 2^{-C_s \cdot t} \right) \pm 2 \cdot 2^{-C_s \cdot t \cdot k}\right).  \label{eq:kminwise_expansion} 
    \end{align}
\end{claim}

In \eqref{eq:kminwise_expansion}, the product $\prod_{i \notin J} \left(\Pr_{\sigma \sim U}[\min \sigma(B_i) > \theta] \pm 2 \cdot 2^{-C_s \cdot t} \right)$ replaces dependent seeds $s_i$ for $i \in [\ell]\setminus J$ by independent seeds. The term $2 \cdot 2^{-C_s \cdot t}$ comes from the error of the extractor and the error of $\PRG_2$. Moreover, the last error term $2 \cdot 2^{-C_s \cdot t \cdot k}$ in \eqref{eq:kminwise_expansion} is similar to the error term \eqref{eq:after_NZ_PRG2} in \Cref{clm:dependent_seeds}, which is introduced by $\PRG_1$.

One more remark is that this shows the sum of $t$-wise independence and the Nisan-Zuckerman PRG could fool conditional events like \eqref{eq:condition_event}. The key of  \eqref{eq:kminwise_expansion} is to fool event $\Pr[\max \sigma(Y)=\theta]$ with a multiplicative error $2\cdot 2^{-C_s \cdot t \cdot k}$.

We split \eqref{eq:kminwise_expansion} into two parts
\begin{align}
         & \Pr_{\sigma : (C_e+1) \cdot k \text{-wise}}[\max \sigma(Y)=\theta \wedge \min \sigma(B_J) > \theta ] \cdot  \prod_{i \notin J} \left(\Pr_{\sigma \sim U}[\min \sigma(B_i) > \theta] \pm 2 \cdot 2^{-C_s \cdot t} \right) 
        \label{eq:kminwise_eq1}\\
        & \pm \Pr_{\sigma : (C_e+1) \cdot k \text{-wise}}[\max \sigma(Y)=\theta \wedge \min \sigma(B_J) > \theta ] \cdot 2^{-C_s \cdot t \cdot k + 1}. \label{eq:kminwise_eq2}
\end{align}
 
Similar to the proof strategy of \Cref{thm:min_wise_log_seed}, we bound \eqref{eq:kminwise_eq1} and \eqref{eq:kminwise_eq2} separately.

\begin{claim}\label{clm:kminwise_PRG_error}
    The summation of \eqref{eq:kminwise_eq2},
    \[
    \sum_{\theta \in [M]} \Pr_{\sigma : (C_e+1) \cdot k \text{-wise}}[\max \sigma(Y)=\theta \wedge \min \sigma(B_J) > \theta ] \cdot 2^{-C_s \cdot t \cdot k + 1} \le 2^{-2 C \cdot t} /|X|^k.
    \] 
\end{claim}

Then we bound \eqref{eq:kminwise_eq1} by the following claim, which shows its summation is an approximation of $k$-min-wise hash with a small multiplicative error. 

\begin{claim}\label{clm:kminwise_main_term}
    The summation of \eqref{eq:kminwise_eq1} over $\theta$,
    \[\sum_{\theta \in [M]} \Pr_{\sigma : (C_e+1) \cdot k \text{-wise}}[\max \sigma(Y)=\theta \wedge \min \sigma(B_J) > \theta ] \cdot  \prod_{i \notin J} \left(\Pr_{\sigma \sim U}[\min \sigma(B_i) > \theta] \pm 2 \cdot 2^{-C_s \cdot t} \right),\] 
    equals $
    (1 \pm 2^{-2 C \cdot t}) \cdot \Pr_{\sigma \sim U}[\max \sigma(Y) < \min \sigma(X \setminus Y)]
    $.
\end{claim}

For completeness, we show the proof of \Cref{clm:expand} in \Cref{sec:proof_clm_expand}. Then we defer the proofs of \Cref{clm:kminwise_PRG_error} and \Cref{clm:kminwise_main_term} to \Cref{sec:proof_error_kminwise} and \Cref{sec:proof_kminwise_mainterm}. We are ready to finish the proof of \Cref{thm:k_min_wise} here.

\begin{proof}[Proof of \Cref{thm:k_min_wise}]
    We rewrite $\Pr_{h \sim \mathcal{H}}[\max h(Y)< \min h(X \setminus Y)]$ as \eqref{eq:all_events}. Then we apply \Cref{clm:expand} to each term $\Pr_{h \sim \mathcal{H}}\left[ \max h(Y) = \theta \wedge \min h(X \setminus Y) > \theta \right]$. 

    Next we decompose the bound \eqref{eq:kminwise_expansion} in \Cref{clm:expand} into two terms: \eqref{eq:kminwise_eq1} and \eqref{eq:kminwise_eq2}. Finally, \Cref{clm:kminwise_PRG_error} shows \eqref{eq:kminwise_eq2} is a small multiplicative error and \Cref{clm:kminwise_main_term} shows \eqref{eq:kminwise_eq1} is an approximation with multiplicative error $2^{- 2 C \cdot t}$.
\end{proof}    

\subsection{Proof of \Cref{clm:expand}}\label{sec:proof_clm_expand}

We enumerate the non-empty subset $Y' \subseteq Y$ with $h(Y')=\theta$ to rewrite $\max h(y)=\theta$ as $(h(Y')=\theta \wedge \max h(Y \setminus Y')<\theta)$. So our goal is to bound
\begin{align}
     & \Pr_{h \sim \mathcal{H}}[ \max h(Y)=\theta \wedge \min h(X \setminus Y)> \theta] \notag \\
     = & \sum_{\varnothing \ne Y' \subseteq Y} \Pr_{h \sim \mathcal{H}}[h(Y')=\theta \wedge \max h(Y\setminus Y')<\theta \wedge \min h(X \setminus Y)>\theta].  \label{eq:pf_all_events}
\end{align}
Let us consider each probability for a fixed $Y'$.

For convenience, we define $I_i(h_0,z_i)$ to denote the indicator of that hash $h(x) = h_0(x) + \sigma_i(x)$ for $\sigma_i=\PRG_2(z_i)$ satisfies all conditions in \eqref{eq:pf_all_events} for $x \in X$ mapped to bucket $i$, i.e., $h(x)=\theta$ for $x \in Y'$ with $g(x)=i$, $h(x)<\theta$ for $x \in Y \setminus Y'$ with $g(x)=i$, and $h(x)>\theta$ for $x \in X \setminus Y$ with $g(x)=i$. Then we can rewrite $\Pr_{h \sim \mathcal{H}}\left[ h(Y') = \theta \wedge \max h(Y \setminus Y') < \theta \wedge \min h(X \setminus Y) > \theta \right]$ as 
\begin{equation}
    \E_{h_0, w, (s_1, \ldots, s_\ell) \sim \PRG_1 : z_i = \Ext(w,s_i)} 
    \left[\prod_{i \in [\ell]} I_i(h_0,z_{i})\right]. \label{eq:sum_buckets}
\end{equation}

We use the analysis of the Nisan-Zuckerman PRG implicitly. In the first step, we apply $h_0$ to calculating $\E\left[ \prod_{j \in J} I_j \right]$. In this calculation, we fix $z_J := (z_{j_1}, \ldots, z_{j_{k'}})$ and their corresponding functions $\sigma_J:=(\sigma_{j_1},\ldots,\sigma_{j_{k'}})$. Recall that $h_0$ is $(C_e + 1) \cdot k$-wise independent and $|Y| \le k$ are mapped to $J := \{j_1,\ldots,j_{k'}\}$ under $g$. Therefore
\begin{align}
      & \E_{h_0, w, s \sim \PRG_1}\left[ \prod_{j \in J} I_j \right] = \E_{w, s \sim \PRG_1}\left[ \E_{h_0}\left[ \prod_{j \in J} I_j\ \middle|\ z_J, \sigma_J\right] \right]  \notag \\
= & \Pr_{\sigma : (C_e+1)k \text{-wise}}[\sigma(Y')=\theta \wedge \max \sigma(Y\setminus Y')<\theta \wedge \min \sigma(B_J)>\theta]. \label{eq:events_in_J}
\end{align}

Now we fix $h_0$ and $z_J$ such that $\prod_{j \in J} I_j = 1$ (otherwise it contributes $0$ to \eqref{eq:all_events}) and analyze the conditional expectation:
\begin{equation}\label{eq:before_PRG_kmin}
    \E_{h_0, w, s \sim \PRG_1}\left[\prod_{i \notin J} I_i\ \middle|\ \prod_{j \in J}I_j = 1 \right].
\end{equation}

Similar to the proof in \Cref{thm:min_wise_log_seed}, we use the property of $\PRG_1$ to replace $s_1,\ldots,s_\ell$ by independent seeds. Let $\alpha_J = (\alpha_{j_1}, \ldots, \alpha_{j_{k'}}) \in \{0, 1\}^{C_e \cdot t \times J}$ denote the enumeration of $s_J=(s_{j_1},\ldots,s_{j_{k'}})$. We expand the conditional expectation as
\begin{align*}
      & \E_{h_0, w, s \sim \PRG_1}\left[\prod_{i \notin J} I_i\ \middle|\ \prod_{j \in J}I_j = 1 \right] = \sum_{\alpha_J} \E_{h_0, w, s \sim \PRG_1}\left[ \mathbf{1}(s_J = \alpha_J) \cdot \prod_{i \notin J} I_i\ \middle|\ \prod_{j \in J}I_j = 1 \right] \\
    = & \sum_{\alpha_J} \E_{h_0, w}\left[ \Pr_{s \sim \PRG_1}\left[ s_J = \alpha_J\ \middle|\ \prod_{j \in J} I_j = 1 \right] \cdot \E_{s \sim \PRG_1}\left[ \prod_{i \notin J} I_i\ \middle|\ s_J = \alpha_J,  \prod_{j \in J} I_j = 1 \right]  \right].
\end{align*}

\noindent For each fixed $w$ and $h_0$, $\prod_{i \notin J} I_i$ is a combinatorial rectangle, whose inputs $(s_i)_{i \notin J}$ are in $\{0, 1\}^{C_e \cdot t}$. Hence we apply $\PRG_1$ to $\E_{s \sim \PRG_1}\left[ \prod_{i \notin J} I_i\ \middle|\ s_J = \alpha_J \right]$. Observed that $\prod_{j \in J} I_j = 1$ is determined given $w$, $h_0$, and $s_J=\alpha_J$, which could be neglected here.

Because $|s_J| = k' \cdot C_e \cdot t \le C_e \cdot t \cdot k$, by the same argument of \Cref{fact:PRG_CR_LOG}, it follows that
$$
\E_{s \sim \PRG_1}\left[ \prod_{i \notin J} I_i\ \middle|\ s_J = \alpha_J \right] = \prod_{i \notin J} \E_{s_i \sim U}[I_i] \pm 2^{-C_s \cdot t \cdot k}.
$$

\noindent This simplifies \eqref{eq:before_PRG_kmin} to
\begin{align}
    \sum_{\alpha_J} \Pr_{h_0, w, s \sim \PRG_1}\left[ s_J = \alpha_J\ \middle|\ \prod_{j \in J} I_j = 1 \right] \cdot  \left( \E_{h_0, w, (s_i)_{i \notin J} \sim U}\left[ \prod_{i \notin J} I_i \ \middle|\ \prod_{j \in J} I_j = 1 \right] \pm 2^{-C_s \cdot t \cdot k} \right). \label{eq:error_kminwise_conditioning}
\end{align}

We will bound $\E_{h_0, w, (s_i)_{i \notin J} \sim U}\left[ \prod_{i \notin J} I_i \ \middle|\ \prod_{j \in J} I_j = 1 \right]$ for any fixed $s_J=\alpha_J$ conditioned with $\prod_{j \in J} I_j=1$. We apply the Nisan-Zuckerman analysis because seeds $s_i$ are independent here. Since $|z_J| = k' \cdot C_e \cdot \log N \le k \cdot C_e \cdot \log N \le 0.1 \cdot |w|$ given $|w| = 10k \cdot C_e \cdot \log N$, the min-entropy of $w$ is at least $0.8 \cdot |w|$ with probability $1-2^{-0.1 \cdot |w|}$ after fixing $z_J$. So we assume the min-entropy of $w$ is at least $0.8 |w|$ conditioned on that $h_0$ and $z_J$ will satisfy $\prod_{j \in J} I_j=1$.

Then for $i \notin J$, $z_i = \Ext(w,s_i)$ with $s_i \sim U$ is $\NZerr$-close to the uniform distribution, which implies $\sigma_i$ is $\NZerr$-close to the uniform distribution in $\PRG_2$. We repeat this argument for every $i \notin J$ and obtain
\begin{equation}\label{eq:rest_events}
    \E_{h_0, w, (s_i)_{i \notin J} \sim U}\left[ \prod_{i \notin J} I_i \ \middle|\ \prod_{j \in J} I_j = 1 \right] = \prod_{i \notin J} \left( \E_{\sigma_i \sim \PRG_2}[I_i] \pm \NZerr \right) \pm (\ell - k') \cdot 2^{-0.2 \cdot |w|}.  
\end{equation}

Moreover, $\E_{\sigma_i \sim \PRG_2}[I_i]$ is equal to $\Pr_{\sigma \sim U}[\min \sigma(B_i) > \theta] \pm 2 \cdot 2^{- C_s \cdot t}$ by the definition of $\sigma_i = \PRG_2(z_i)$ for a PRG with error $2^{- C_s \cdot t}$. 
    
The additive term $(\ell - k') \cdot 2^{-0.2 \cdot |w|}$ is the union bound for the min-entropy of $w$ is less than $0.6 \cdot |w|$ after conditioned previous indicators are $1$. Since $t := \frac{\log N}{\log \log N}$ and $|w| = 10 k \cdot C_e \cdot \log N$, we combine it with the error in \eqref{eq:error_kminwise_conditioning} as $2 \cdot 2^{- C_s \cdot t \cdot k}$. Thus \eqref{eq:error_kminwise_conditioning} becomes
\begin{align}
& \sum_{\alpha_J} \Pr_{h_0, w, s \sim \PRG_1}\left[ s_J = \alpha_J\ \middle|\ \prod_{j \in J} I_j = 1 \right] \notag \\
 & \quad \ \cdot \left(\prod_{i \notin J} \left( \Pr_{\sigma \sim U}[\min \sigma(B_i) > \theta] \pm \NZerr \pm 2^{-C_s \cdot t}\right) \pm 2^{-C_s \cdot t \cdot k} \pm 2^{-0.1 \cdot |w|} \pm (\ell-k')\cdot 2^{-0.2 \cdot |w|} \right) \notag \\
   = & \prod_{i \notin J} \left( \Pr_{\sigma \sim U}[\min \sigma(B_i) > \theta] \pm 2 \cdot 2^{-C_s \cdot t}\right) \pm 2\cdot 2^{-C_s \cdot t \cdot k}. \label{eq:simplied_prod_kminwise}
\end{align}

We complete this proof by plugging \eqref{eq:events_in_J} for $\E\left[ \prod_{j \in J} I_j \right]$ and \eqref{eq:simplied_prod_kminwise} for $\E\left[\prod_{i \notin J} I_i\ \middle|\ \prod_{j \in J} I_j \right]$.
\begin{align*}
    & \E\left[ \prod_{i \in [\ell]} I_i \right] = \E\left[ \prod_{j \in J} I_j \right] \cdot \E\left[ \prod_{i \notin J} I_i\ \middle|\ \prod_{j \in J} I_j=1 \right]\\
    = & \Pr_{\sigma : (C_e+1)k \text{-wise}}[\sigma(Y')=\theta \wedge \max \sigma(Y\setminus Y')<\theta \wedge \min \sigma(B_J)>\theta] \\
     & \cdot \left( \prod_{i \notin J} \left( \Pr_{\sigma \sim U}[\min \sigma(B_i) > \theta] \pm 2 \cdot 2^{-C_s \cdot t}\right) \pm 2\cdot 2^{-C_s \cdot t \cdot k} \right).
\end{align*}
Moreover, by \eqref{eq:pf_all_events},
$$
\sum_{Y'} \Pr_{\sigma : (C_e+1)k \text{-wise}}[\sigma(Y')=\theta \wedge \max \sigma(Y\setminus Y')<\theta \wedge \min \sigma(B_J)>\theta]
$$
$$
= \Pr_{\sigma : (C_e+1)k \text{-wise}}[\max \sigma(Y)=\theta \wedge \min \sigma(B_J)>\theta],
$$
which finishes this proof.

\subsection{Proof of \Cref{clm:kminwise_PRG_error}}\label{sec:proof_error_kminwise}

If $|X| < 2^{0.5 C_s \cdot t + 1}$, then the last factor $2^{- C_s \cdot t \cdot k + 1}$ implies that this summation is at most $2^{-0.5 C_s \cdot t \cdot k}/|X|^k$.

Otherwise $|X| \ge 2^{0.5 C_s \cdot t + 1}$ such that \Cref{lem:allocation_k_min_wise} implies that all $B_i$ has $|B_i| = (1 \pm 0.1) \cdot |X \setminus Y| / \ell$. So $|B_J| = (1 \pm 0.1) \cdot k' \cdot |X \setminus Y| / \ell$. 

If we neglect the $2^{-C_s \cdot t \cdot k + 1}$ factor and consider $\sum_{\theta \in [M]} \Pr_{\sigma : (C_e+1) \cdot k \text{-wise}}[\max \sigma(Y)=\theta ~\wedge~ \min \sigma(B_J) > \theta ]$, this is the exact probability of $B_J \cup Y$ satisfying the $k$-min-wise hash condition under $(C_e + 1) \cdot k$-wise independence. Feigenblat, Porat and Shiftan \cite{minwise_FPS11} have shown this bounded by $2 k! / |B_J|^k$ when $C_e$ is large enough. We state their results for completeness.

\begin{theorem}[Theorem 1.1 in \cite{minwise_FPS11}]
    There exists a constant $c$ such that for any $\varepsilon > 0$, any $c \cdot \left( k \log \log (1/\varepsilon) + \log (1 / \varepsilon) \right)$-wise independent function from $[N]$ to $[M]$ is a $k$-min-wise hash family of error $\varepsilon$ when $M = \Omega(N /\varepsilon)$. 
\end{theorem}

So we simplify the summation as
$$
\frac{O(1)}{{|B_J \cup Y| \choose |Y|}} \cdot 2^{-C_s \cdot t \cdot k + 1} \le \frac{2 k!}{|B_J|^k} \cdot 2^{-C_s \cdot t \cdot k + 1} \le 2k! \cdot 2^{-C_s \cdot t \cdot k + 1} \cdot \left( \frac{\ell}{0.9 \cdot k' \cdot |X \setminus Y|} \right)^k \le 2^{-0.5 C_s \cdot t \cdot k} / |X|^k,
$$

\noindent given $\ell=2^{t}$.

\subsection{Proof of \Cref{clm:kminwise_main_term}}\label{sec:proof_kminwise_mainterm}

First of all, it would be more convenient to rewrite the summation as 
\begin{align*}
    & \sum_{\theta \in [M]} \Pr_{\sigma : (C_e+1) \cdot k \text{-wise}}[\max \sigma(Y)=\theta \wedge \min \sigma(B_J) > \theta ] \cdot  \prod_{i \notin J} \left(\Pr_{\sigma \sim U}[\min \sigma(B_i) > \theta] \pm 2 \cdot 2^{-C_s \cdot t} \right) \\
    = & \sum_{\theta \in [M]} \Pr_{\sigma \sim U}[\max \sigma(Y)=\theta] \cdot \Pr_{\sigma : C_e \cdot k \text{-wise}}[\min \sigma(B_J) > \theta ] \cdot  \prod_{i \notin J} \left(\Pr_{\sigma \sim U}[\min \sigma(B_i) > \theta] \pm 2 \cdot 2^{-C_s \cdot t} \right) \\
    = & \sum_{\theta \in [M]} \frac{\theta^k-(\theta-1)^k}{M^k} \cdot \Pr_{\sigma : C_e \cdot k \text{-wise}}[\min \sigma(B_J) > \theta ] \cdot  \prod_{i \notin J} \left(\Pr_{\sigma \sim U}[\min \sigma(B_i) > \theta] \pm 2 \cdot 2^{-C_s \cdot t} \right).
\end{align*}

This proof considers the three cases in \Cref{lem:allocation_k_min_wise} according to the size of $|X|$. Recall that we have assumed $k = \log^{O(1)} N$. Set $\varepsilon := 2 \cdot 2^{-C_s \cdot t}$ for simplicity.

\paragraph{The first case of $|X| \le \ell^{0.9}$.}
\Cref{lem:allocation_k_min_wise} implies $|B_i| \le C_g + 10 \cdot \frac{k \log |X|}{\log N/ \log \log N},\ \forall i \in [\ell]$ and $|B_J| \le C_g \cdot k$. Thus the middle term $\Pr_{\sigma: C_e k\text{-wise}}[\min \sigma(B_J)>\theta]$ has no error. Let us focus on the product
$$
\prod_{i \notin J} \left(\Pr_{\sigma \sim U}[\min \sigma(B_i)>\theta] \pm 2 \cdot 2^{-C_s t} \right) = \prod_{i \notin J} \left( (1 - \theta / M)^{|B_i|} \pm \varepsilon
\right).
$$

Without loss of generality, we assume $J=\{\ell-k'+1,\ldots,\ell\}$ and $|B_1| \ge |B_2| \ge \cdots \ge |B_{\ell-k'}|$. When $\theta$ is small, say $(1-\theta/M)^{|B_1|} \ge 2^{-0.5 C_s \cdot t}$, this product has a small multiplicative error:
\begin{align*}
      & \prod_{i \notin J} \left( (1 - \theta / M)^{|B_i|} \pm \varepsilon \right) = \prod_{i \notin J} (1 - \theta / M)^{|B_i|} \cdot (1 \pm 2^{-0.5 C_s \cdot t}) \\
    = & ( 1 \pm \ell \cdot 2^{-0.5 C_s \cdot t} ) \cdot \prod_{i \notin J} (1 - \theta / M)^{|B_i|} = (1 \pm 2^{-(0.5 C_s - 1) \cdot t}) \cdot \prod_{i \notin J} (1 - \theta / M)^{|B_i|}.
\end{align*}

Otherwise $(1 - \theta / M)^{|B_1|} < 2^{-0.5 C_s \cdot t}$ is already small enough such that we only need to give a tail bound. Let $S := 0.5 k \cdot \log \log N$ and $B_{\notin J} := \bigcup_{i = 1}^{\ell - k'} B_i$ for convince.
\begin{enumerate}
    \item The easy case is $|B_{\notin J}| \ge |B_1| \cdot k \log \log N$, which tells that $(1 - \theta/M)^{|B_{\notin J}|} \le (2^{-0.5 C_s \cdot t})^{k \log \log N} = N^{-0.5 C_s \cdot k}$ is negligible. We show $\prod_{i \notin J} \left( (1 -\theta / M)^{|B_i|} \pm \varepsilon \right)$ is also negligible in this case.
    \begin{itemize}
        \item If $(1 - \theta / M)^{|B_S|} \le 2^{-t - 1}$, $\prod_{i = 1}^{S} \left( (1 - \theta / M)^{|B_i|} \pm \varepsilon \right) \le N^{-0.5 k}$ is sufficiently small.
    
        \item If not, then $\prod_{i = S + 1}^{\ell - k'}\left( (1 - \theta / M)^{|B_i|} \pm \varepsilon \right) \le 2 \cdot \prod_{i = S + 1}^{\ell - k'} ( 1 -\theta / M)^{|B_i|} = N^{-\Omega(k)}$ is sufficiently small.
    \end{itemize}

    \item When $|B_{\notin J}| \le |B_1| \cdot k \log \log N = \log^{O(1)} N$, this further implies $|B_i| \le C_g + 10 \cdot \frac{k (\log \log N)^2}{\log N}$.
    \begin{itemize}
        \item If $10 \cdot \frac{k (\log \log N)^2}{\log N} \ge C_g$, then $k \ge 0.1 C_g \cdot \frac{\log N}{(\log \log N)^2}$ such that the number of non-empty buckets is at least $\frac{k}{20\frac{k (\log \log N)^2}{\log N}}=\frac{\log N}{20 (\log \log N)^2}$. So we could prove $\prod_{i = S + 1}^{\ell - k'} \left( (1 - \theta / M)^{|B_i|} \pm \varepsilon \right)$ is negligible by considering $B_S$ again like the above case.

        \item If $10 \cdot \frac{k (\log \log N)^2}{\log N} < C_g$, then each bucket has at most $2 C_g$ elements. Also we have $k=O\left(\frac{\log N}{(\log \log N)^2}\right)$ and $|X|=O(k \log \log N)=O\left(\frac{\log N}{\log \log N}\right)$ such that $1/|X|^k=2^{-O\left(\frac{\log N}{\log \log N}\right)}$.
        
        From the condition $(1 - \theta / M)^{|B_1|} < 2^{-0.5 C_s \cdot t}$, we have $\theta / M \ge 1 - 2^{-(C_s \cdot t) / (4 C_g)}$. The number of such $\theta$'s is at most $M \cdot 2^{-(C_s \cdot t) / (4 C_g)}$. Using the fact $\frac{\theta^k - (\theta - 1)^k}{M^k} \le k / M$, this implies the additive error is $k \cdot 2^{-(C_s \cdot t) / (4 C_g)} < 2^{-3 C \cdot t}$.
    \end{itemize}
\end{enumerate}

\paragraph{The second case of $|X| \in (\ell^{0.9},\ell^{1.1})$.} The failure probability in \Cref{lem:allocation_k_min_wise} has a negligible impact on the final multiplicative error. Thus we assume that the max-load $\max_{i \in [\ell]} |B_i|$ is bounded by $2 \ell^{0.1}$.

When $\theta / M \le 0.5 \ell^{-0.2}$, $(1 - \theta / M)^{|B_J|} \ge 1 - |B_J| \cdot \theta / M \ge 1 - k \cdot \ell^{-0.1} \ge 0.5$. By the first statement of \Cref{lem:t_wise_ind}, we have
$$
\Pr_{\sigma : \text{$C_e k$-wise}}[\min \sigma(B_J) > \theta] = (1 - \theta / M)^{|B_J|} \pm (|B_J| \cdot \theta / M)^{C_e \cdot k} = (1 - \theta / M)^{|B_J|} \cdot \left( 1 \pm 2 \left( \frac{k}{\ell^{0.1}} \right)^{C_e \cdot k} \right),
$$

\noindent and for $i \notin J$:
$$
(1 - \theta / M)^{|B_i|} \pm \varepsilon = (1 - \theta / M)^{|B_i|} \cdot (1 \pm 2 \varepsilon).
$$

Since $k = \log^{O(1)} N$ and $\ell = 2^t$, which tells that $k \ll \ell$, the multiplicative error $B_J$ is bounded by $2 \left( \frac{k}{\ell^{0.1}} \right)^{C_e \cdot k} \le 2 \ell^{-0.05 C_e \cdot k}$. So
\begin{align*}
      & \Pr_{\sigma : \text{$C_e k$-wise}}[\min \sigma(B_J) > \theta] \cdot \prod_{i \notin J} \left( (1 - \theta / M)^{|B_i|} \pm \varepsilon \right) \\
    = & (1 - \theta / M)^{|X \setminus Y|} \cdot \left( 1 \pm \left( \frac{4}{\ell^{0.05 C_e \cdot k}} + 4 (\ell - k') \cdot \varepsilon \right) \right) \\
    = & (1 - \theta / M)^{|X \setminus Y|} \cdot (1 \pm 2^{-0.5 C_s \cdot t}).
\end{align*}

When $\theta / M > 0.5 \ell^{-0.2}$, we only need to estimate $\sum_{\theta > 0.5 M \cdot \ell^{-0.2}} \frac{\theta^k - (\theta - 1)^k}{M^k} \cdot \prod_{i \notin J} \left( (1 - \theta / M)^{|B_i|} + \varepsilon \right)$. Consider the $S := k \cdot C' \cdot \log \log N$ largest number among $\{|B_i|\}_{i \notin J}$, denoted by $b$.
\begin{enumerate}
    \item If $(1 - \theta / M)^b \ge \varepsilon \cdot \ell$, then
    \begin{align*}
          & \sum_{\theta > 0.5 M \cdot \ell^{-0.2} \wedge (1 - \theta / M)^b \ge \varepsilon \cdot \ell} \frac{\theta^k - (\theta - 1)^k}{M^k} \cdot \prod_{i \notin J} \left( (1 - \theta / M)^{|B_i|} + \varepsilon \right) \\
        \le & \sum_{\theta > 0.5 M \cdot \ell^{-0.2} \wedge (1 - \theta / M)^b \ge \varepsilon \cdot \ell} \frac{\theta^k - (\theta - 1)^k}{M^k} \cdot \prod_{i \notin J : |B_i| > b} \left( (1 - \theta / M)^{|B_i|} + \varepsilon \right) \\
        \le & \sum_{\theta > 0.5 M \cdot \ell^{-0.2} \wedge (1 - \theta / M)^b \ge \varepsilon \cdot \ell} \frac{\theta^k - (\theta - 1)^k}{M^k} \cdot e \cdot (1 - \theta / M)^{|X \setminus Y| - (S + k') \cdot 2 \ell^{0.1}} \\
        \le & \exp(1 - 0.9 |X| \cdot 0.5 \ell^{-0.2}) \le \exp(-\ell^{0.6}).
    \end{align*}

    \item If $(1 - \theta / M)^b < \varepsilon \cdot \ell$, then
    \begin{align*}
          & \sum_{\theta > 0.5 M \cdot \ell^{-0.2} \wedge (1 - \theta / M)^b < \varepsilon \cdot \ell} \frac{\theta^k - (\theta - 1)^k}{M^k} \cdot \prod_{i \notin J} \left( (1 - \theta / M)^{|B_i|} + \varepsilon \right) \\
        \le & \sum_{\theta > 0.5 M \cdot \ell^{-0.2} \wedge (1 - \theta / M)^b < \varepsilon \cdot \ell} \frac{\theta^k - (\theta - 1)^k}{M^k} \cdot (\varepsilon \cdot \ell + \varepsilon)^S \\
        \le & 2^{2 S - C' \cdot (C_s - 1) \cdot t \cdot S} = 1 / N^{O(k)}.
    \end{align*}
\end{enumerate}

\paragraph{The third case of $|X| \ge \ell^{1.1}$.} Again, the proof of this case is the same as the second case, and is thus omitted.
\section{Extractors}\label{sec:extractor}

We restate \Cref{lem:extractor_unif} here and finish its proof in this section.
Different from previous works, this randomness extractor has an extra property: $\Ext(U_n, s) = U_m$ for any fixed seed $s$.

\begin{lemma}\label{lem:extractor_unif_2}
    Given any $n$ and $k < n$, for any error $\varepsilon$, there exists a randomness extractor $\Ext : \{0, 1\}^n \times \{0, 1\}^d \to \{0, 1\}^m$ with $m = k/ 2$ and $d = O(\log (n / \varepsilon))$. Moreover, $\Ext$ satisfies an extra property: $\Ext(U_n, s) = U_m$ for any fixed seed $s$. 
\end{lemma}

We show how to design explicit extractors for \Cref{lem:extractor_unif_2} in the rest of this section. Our construction will be based on the linear form of the lossless condenser in \cite{GUV,CI_GUV_linear}. 

\begin{definition}
    $h:\mathbf{F}_2^n \times \mathbf{F}_2^d \to \mathbf{F}_2^m$ is a $(k,\varepsilon)$-lossless condenser if for any random source $X$ of min-entropy at least $k$, the distribution $(Z, h(X,Z))$ is $\varepsilon$-close to some distribution with min-entropy at least $k+d$, where $Z$ is an independent seed uniformly distributed in $\mathbf{F}_2^d$.
\end{definition}

We strengthen the linear seeded condensers in \cite{GUV,CI_GUV_linear} such that they are surjective for every seed. We state the main properties of lossless condensers as follows, which reformulates Corollary 38 of \cite{CI_GUV_linear} with an extra surjective guarantee.

\begin{lemma}\label{lem:linear_condenser}
    Let $\alpha>0$ be any constant. For any $n$, any $k<n$,  and $\varepsilon>0$, there exists an explicit $(k,\varepsilon)$-lossless condenser $h : \mathbf{F}_2^n \times \mathbf{F}_2^d \to \mathbf{F}_2^m$ with $d \le (1 + 1 / \alpha) \cdot \log (n k / \varepsilon) + O(1)$ and $m \le d+(1+\alpha)k$. Moreover, this lossless condenser satisfies the following two properties:
    \begin{enumerate}
        \item For every seed $s \in \mathbf{F}_2^d$, $h(x,s)$ is a linear function on $x$.
        \item For every seed $s \in \mathbf{F}_2^d$, $h(x,s)$ is surjective such that $h(U_n,s)=U_m$.
    \end{enumerate}
\end{lemma}

\begin{proof}
    Corollary 38 in \cite{CI_GUV_linear} provides a linear lossless condenser with the above parameters although it is not necessarily surjective. In particular, its construction provides a condenser \emph{for every $n$ and every $k<n$} as long as there exists an irreducible univariate polynomial of degree $n$ in some finite extension field of $\mathbf{F}_2$. Irreducible polynomials of degree $n$ always exist because the number of irreducible polynomials of degree $n$ in the finite field $\mathbf{F}_q$ is $\frac{1}{n}\sum_{d \mid n}\mu(n/d)q^d$ greater than 0 (by the Gauss formula). Moreover, there are efficient algorithms \cite{Shoup_irreducible} to find such an irreducible polynomial.
    
    Now we modify its construction to satisfy the second property.

    Specifically, for every seed $s$, if $h(x,s)$ is not surjective, we expand it into a surject function on $\mathbf{F}_2^m$. Namely, if $h(x,s)=M_s \cdot x$ for a matrix $M_s \in \mathbf{F}_2^{m \times n}$ whose rank is not $m$, then we replace every linearly dependent row in $M_s$ by an independent vector. Let $M'_s$ be the new matrix and $h'(x,s)=M'_s \cdot x$ be the new condenser after guaranteeing the new matrix is of rank $m$.

    Hence $h'(\cdot,s)$ is surjective and $h'(\cdot,s)$ is still linear. Moreover, the min-entropy of $h'(X,s)$ is not less than the min-entropy of $h(X,s)$ for any fixed $s$. Thus $H_{\infty}(Z,h'(X,Z)) \ge H_{\infty}(Z,h(X,Z))$.
\end{proof}

Similar to \Cref{lem:linear_condenser}, we modify the classical leftover hash lemma to construct a linear extractor $E$ of optimal error such that $E(U_n,s)=U_m$. 

\begin{claim}\label{clm:new_leftover}
    For any $n$, any $k<n$, and $m<k$, there exists a $(k,2 \cdot 2^{\frac{m-k}{2}})$-strong extractor $E:\{0,1\}^n \times \{0,1\}^{n-1} \rightarrow \{0,1\}^m$ such that
    \begin{enumerate}
        \item $E(x,s)$ is a linear function of $x$ for any seed $s$. 
        \item $E(U_n,s)$ is surjective such that $E(U_n,s)=U_m$ for any seed $s$.
    \end{enumerate}
\end{claim}

\begin{proof}
    We consider the extension field $\mathbf{F}_{2^n}$ of size $2^n$ and view each element $\alpha \in \mathbf{F}_{2^n}$ as a vector in $\{0,1\}^n$. 
    For every seed $s \in \{0,1\}^{n-1}$, we pick a distinct \emph{non-zero} element $y_s \in \mathbf{F}_{2^n}$ and define $E(x,s):=(x \cdot y_s)_{m}$ to output the first $m$ bits of the product $x \cdot y_s \in \mathbf{F}_{2^n}$. This guarantees that $E(U_n,s)=U_m$ for any seed $s$ (since $y_s \neq 0$).
    
    Then we bound its error by $2 \cdot 2^{\frac{m-k}{2}}$. The standard leftover hash lemma shows that $\big( y,(xy)_{m} \big) \approx_{2^{\frac{m-k}{2}}} (U_n,U_m)$ when $x \sim X$ of min-entropy $k$ and $y \sim \mathbf{F}_{2^n}$. However, the support size of $y$ is $2^{n-1}$ instead of $2^n$ in our construction. But this will only increase the error by a factor of 2 (via the Markov inequality).   
\end{proof}

The two extra properties in \Cref{clm:new_leftover} are identical to the properties in \Cref{lem:linear_condenser}. These properties help us to design an extractor such that $E(U_n,s)=U_m$. The second property guarantees that the output is uniformly distributed over $\mathbf{F}_2^m$ for any seed $s$ whenever we apply \Cref{lem:linear_condenser} or \Cref{clm:new_leftover} to a uniform source. The first property further shows that $U_n \mid E(U_n,s)$ is still a uniform random source in the dual space of $h(\cdot,s)$ of dimension $n-m$. This allows our construction to continue this type of randomness extraction and condensation.

While our construction of \Cref{lem:extractor_unif_2} follows the same outline of Theorem 5.12 in \cite{GUV}, one \emph{subtle difference} is that after every application of \Cref{lem:linear_condenser} or \Cref{clm:new_leftover}, we replace the random source $X$ by its projection onto the dual space of the linear map.

\begin{proof}[Proof of \Cref{lem:extractor_unif_2}]
    The difference between our construction and Theorem 5.12 in \cite{GUV} are (1) we replace every operation of leftover hash lemma by \Cref{clm:new_leftover} and replace every operation of lossless condensers by \Cref{lem:linear_condenser}; (2) after each operation, we project the random source to the dual space of that operation. Specifically, let $X_0 \in \mathbf{F}_2^n$ be the initial random source and $X_i \in \mathbf{F}_2^{n_i}$ be the random source after applying $i$ times \Cref{clm:new_leftover} and \Cref{lem:linear_condenser}. For example, suppose the $(i+1)$-th operation applies the lossless condenser $h_i:\mathbf{F}_2^{n_i} \times \mathbf{F}_2^{d_i} \rightarrow \mathbf{F}_2^{m_i}$ with seed $s_i$ in \Cref{lem:linear_condenser}, say $h_i(X_{i},s_i):=A \cdot X_{i}$ for some full rank matrix $A \in \mathbf{F}_2^{m_i \times n_i}$. Then we set the next random source to be $X_{i+1}:=A^{\bot} \cdot X_i$ where  $A^{\bot} \in \mathbf{F}_2^{(n_i-m_i) \times n_i}$ is the dual of $A$. This works because \Cref{lem:linear_condenser} and \Cref{clm:new_leftover} hold for any $n$ and any $k<n$.

    To prove $E(U_n,s)=U_m$, we use the two extra properties in \Cref{lem:linear_condenser} and \Cref{clm:new_leftover}. Our modification guarantees that every $X_i$ is uniform (by inductions).

    The analysis of the error follows the same argument of \cite{GUV}. A small difference is to bound the min-entropy $X_{i+1}$ given $h_i(X_i,s_i)$. Since $h_i(\cdot,s_i)$ is linear, $X_{i+1}$ is the exact distribution of $X_i$ conditioned on $h_i(X_i,s_i)$ such that $H_{\infty} \big( X_{i+1} \big)=H_{\infty}\big( X_i\mid h_i(X_i,s_i) \big)$.
\end{proof}

\section*{Acknowledgements}

We appreciate anonymous reviewers for their helpful comments. We are also grateful to  Parikshit Gopalan for pointing out some typos in an earlier version of this paper.

\bibliographystyle{alpha} 
\bibliography{rand}

\appendix

\section{Connection between Min-wise Hash and PRG for Combinatorial Rectangles}\label{sec:hash_poly_error}

Here we describe the connection between min-wish hash families and pseudo-random generators for combinatorial rectangles \cite{SSZZ00, GY15}.

As mentioned earlier $M=\Omega(N)$, we assume $\Pr_{h \sim U}[h(y) < \min h(X \setminus y)] \ge 1 / (2 |X|)$ and $\Pr_{h \sim U}[\max h(Y) < \min h(X \setminus Y)] \ge 1/ (2 {|X| \choose k})$ for $Y \subseteq X$ of size $k$ in this work.

\begin{lemma}\label{lem:reduction}
    Let $G:\{0,1\}^{s} \rightarrow [M]^N$ be a PRG that fools combinatorial rectangles within additive error $\delta$. Then $G$ provides a min-wise hash family of size $2^s$ and error $2NM \cdot \delta$ and a $k$-min-wise hash family of size $2^s$ and error $\frac{N^k}{k!} \cdot 4M  \delta$.
\end{lemma}

\begin{proof}
    Let $x_1, \ldots, x_n$ be the elements in $X \setminus y$. We rewrite $\Pr_{h \sim \mathcal{H}}[h(y) < \min h(X \setminus y)]$ as
    $$
    \sum_{\theta \in [M]} \Pr_{h \sim \mathcal{H}}\left[ h(y) = \theta \wedge h(x_1) > \theta \wedge \cdots \wedge h(x_n) > \theta \right].
    $$
    
    \noindent Since $\mathbf{1}(h(y) = \theta \wedge h(x_1) > \theta \wedge \cdots \wedge h(x_n) > \theta)$ is a combinatorial rectangle in $[M]^N$, the additive error of each term is at most $\delta$. Hence the total additive error is $M \cdot \delta$. Then, by the assumption $\Pr_{h \sim U}[h(y) < \min h(X \setminus y)] \ge 1 / (2 N)$, we have
    \begin{align*}
          & \sum_{\theta \in [M]} \Pr_{h \sim \mathcal{H}}\left[ h(y) = \theta \wedge h(x_1) > \theta \wedge \cdots \wedge h(x_n) > \theta \right] \\
        = & \sum_{\theta \in [M]} \Pr_{h \sim U}\left[ h(y) = \theta \wedge h(x_1) > \theta \wedge \cdots \wedge h(x_n) > \theta \right] \pm M \cdot \delta \\
        = & \Pr_{h \sim U}[h(y) < \min h(X \setminus y)] \cdot (1 \pm 2 N M \delta).
    \end{align*}

    Similarly, for the $k$-min-wise hash, we express $\Pr[\max h(Y)=\theta \wedge \min h(X \setminus Y)>\theta]$ as 
    $$
    \Pr[\max h(Y) \le \theta \wedge \min h(X \setminus Y)>\theta] -         \Pr[\max h(Y) \le \theta-1 \wedge \min h(X \setminus Y)>\theta].
    $$
    
    \noindent Then we rewrite $\Pr_{h \sim \mathcal{H}}[\max h(Y) < \min h(X \setminus Y)]$ as
    \begin{align*}
          & \sum_{\theta \in [M]} \Pr_{h \sim \mathcal{H}}[\max h(Y)=\theta \wedge \min h(X \setminus Y)>\theta] \\
        = & \sum_{\theta \in [M]} \left(  \Pr_{h \sim \mathcal{H}}[\max h(Y) \le \theta \wedge \min h(X \setminus Y)>\theta] -         \Pr_{h \sim \mathcal{H}}[\max h(Y) \le \theta - 1 \wedge \min h(X \setminus Y)>\theta] \right)\\
        = & \sum_{\theta \in [M]} \left(     \Pr_{h \sim U}[\max h(Y) \le \theta \wedge \min h(X \setminus Y)>\theta] -         \Pr_{h \sim U}[\max h(Y) \le \theta-1 \wedge \min h(X \setminus Y)>\theta] \right) \pm 2 M \delta \\
        = & \Pr_{h \sim U}[\max h(Y) < \min h(X \setminus Y)] \pm 2 M \delta.
    \end{align*}
    
    \noindent Since we assume $\Pr_{h \sim U}[\max h(Y) < \min h(X \setminus Y)] \ge k! / (2 N^k)$, this $k$-min-wise hash family has an error $\frac{N^k}{k!} \cdot 4 M \delta$.
\end{proof}

Plugging \Cref{thm:PRG_comb_rect} to \Cref{lem:reduction}, Gopalan and Yehudayoff \cite{GY15} had the following results of min-wise and $k$-min-wise hash with small errors.

\begin{theorem}\label{thm:min_wise_small_error}
    Given any $N$ and multiplicative error $\delta$, there is an explicit min-wise hash family of seed length $O\big( \log (NM/\delta) \cdot \log \log (N M / \delta) \big)$.

    More generally, given any $k$, there is an explicit $k$-min-wise hash family of seed length $O\big( (k \log N + \log (1/\delta)) \cdot \log (k\log N + \log (1/\delta)) \big)$.
\end{theorem}
\section{Omitted Calculation in \Cref{sec:second_proof_independent_seeds}}\label{sec:proof_X_big_ind_seed}

Here we complete the calculation omitted for the third case of $|X| > \ell^{1.1}$ in \Cref{sec:second_proof_independent_seeds}. Recall that $\varepsilon = 2 \cdot 2^{-C_s \cdot t}$, and $\mathcal{H}'$ is the hash function family after replacing $s_1, \ldots, s_\ell$ by independent random samples in $\{0, 1\}^{C_e \cdot t}$.

We want to prove that for any $X \subseteq [N]$ with size larger than $\ell^{1.1}$ and any $y \in X$, it holds that $\Pr_{h' \sim \mathcal{H}'}[h'(y) < \min h'(X \setminus y)] = (1 \pm 2^{-2 C \cdot t}) \cdot \Pr_{\sigma \sim U}[\sigma(y) < \min \sigma(X \setminus y)]$.

Again, by \Cref{lem:allocation_small_error}, the failure probability only affects the multiplicative error by at most $|X| / |X|^{3 C} \le 1 / \ell^{1.1 \cdot (3 C - 1)} < 2^{-2.5 C \cdot t}$. 

Hence, we can suppose that in this case all buckets satisfy $|B_i| = (1 \pm 0.1) \cdot |X| / \ell$.

When $\frac{|X|}{\ell} \cdot \frac{\theta}{M} \le \ell^{-0.1}$, there comes $(1 - \theta / M)^{|B_i|} \ge 1 - |B_i| \cdot \theta / M \ge 1 - 1.1 / \ell^{0.1} \ge 0.5,\ \forall i \in [\ell]$. We use the first statement in \Cref{lem:t_wise_ind} to estimate $\Pr_{\sigma : \text{$0.1 C_s$-wise}}[ \min \sigma(B_{j}) > \theta ]$:
\begin{align*}
    \Pr_{\sigma : \text{$0.1 C_s$-wise}}[\min \sigma(B_{j}) > \theta] & = (1 - \theta / M)^{|B_{j}|} \pm \left( |B_{j}| \cdot \dfrac{\theta}{M} \right)^{0.1 C_s} \\
     & = (1 - \theta / M)^{|B_{j}|} \left( 1 \pm 2 \left( \dfrac{1.1}{\ell^{0.1}} \right)^{0.1 C_s} \right) \\
     & = (1 - \theta / M)^{|B_{j}|} \left( 1 \pm O(1 / \ell^{0.01 C_s}) \right).
\end{align*}

For the remaining buckets, we have
$$
(1 - \theta / M)^{|B_i|} \pm \varepsilon = (1 - \theta / M)^{|B_i|} \left( 1 \pm 2 \varepsilon \right),\ \forall i \ne j.
$$

Thus, for relatively small $\theta$, we obtain that
\begin{align*}
     & \frac{1}{M} \sum_{\frac{|X|}{\ell} \cdot \frac{\theta}{M} \le \ell^{-0.1}} \Pr_{\sigma : \text{$0.1 C_s$-wise}}[\min \sigma(B_{j}) > \theta] \cdot \prod_{i \ne j} ((1 - \theta / M)^{|B_i|} \pm \varepsilon) \\
    = & \frac{1}{M} \sum_{\frac{|X|}{\ell} \cdot \frac{\theta}{M} \le \ell^{-0.1}} (1 - \theta / M)^{|X| - 1} \cdot \left( 1 \pm \left( O(1 / \ell^{0.01 C_s}) + 4 (\ell - 1) \varepsilon \right) \right) \\
    = & \frac{1}{M} \sum_{\frac{|X|}{\ell} \cdot \frac{\theta}{M} \le \ell^{-0.1}} (1 - \theta / M)^{|X| - 1} \cdot \left( 1 \pm 2^{-0.001 C_s \cdot t} \right).
\end{align*}

When $\frac{|X|}{\ell} \cdot \frac{\theta}{M} > \ell^{-0.1}$, we need to bound $\frac{1}{M} \sum_{\frac{|X|}{\ell} \cdot \frac{\theta}{M} > \ell^{-0.1}} \prod_{i \ne j} ((1 - \theta / M)^{|B_i|} + \varepsilon)$. Suppose that $j = \ell$, and $1.1 \cdot |X| / \ell \ge |B_1 | \ge |B_2 | \ge \cdots \ge |B_{\ell - 1}|$, and set $S = C' \cdot \log \log N \le \ell - 1$. We split the sum into two cases depending on whether $(1 - \theta / M)^{|B_S|} > \varepsilon \cdot \ell$.
\begin{enumerate}
    \item For $(1 - \theta / M)^{|B_S|} > \varepsilon \cdot \ell$, since $|X| \cdot \theta / M = \ell \cdot \frac{|X|}{\ell} \cdot \frac{\theta}{M} \ge \ell^{0.9}$, we further simplify it as
    \begin{align*}
         & \dfrac{1}{M} \sum_{\frac{|X|}{\ell} \cdot \frac{\theta}{M} > \ell^{-0.1} \wedge (1 - \theta / M)^{|B_S|} > \varepsilon \cdot \ell} \prod_{i \ne j} ((1 - \theta / M)^{|B_i|} + \varepsilon) \\
        \le & \dfrac{1}{M} \sum_{\frac{|X|}{\ell} \cdot \frac{\theta}{M} > \ell^{-0.1} \wedge (1 - \theta / M)^{|B_S|} > \varepsilon \cdot \ell} \prod_{i = S + 1}^{\ell - 1} (1 - \theta / M)^{|B_i|} (1 + 1 / \ell) \\
        \le & \dfrac{1}{M} \sum_{\frac{|X|}{\ell} \cdot \frac{\theta}{M} > \ell^{-0.1} \wedge (1 - \theta / M)^{|B_S|} > \varepsilon \cdot \ell} e \cdot (1 - \theta / M)^{|X| - 1.1 (S + 1) \cdot |X| / \ell} \\
        \le & \exp( 1 - (|X| - 1.1 (S + 1) \cdot |X| / \ell) \cdot \theta / M ) \le \exp(-\ell^{0.8}).
    \end{align*}

    \item For $(1 - \theta / M)^{|B_S|} \le \varepsilon \cdot \ell$, we bound it as
    \begin{align*}
         & \dfrac{1}{M} \sum_{\frac{|X|}{\ell} \cdot \frac{\theta}{M} > \ell^{-0.1} \wedge (1 - \theta / M)^{|B_S|} \le \varepsilon \cdot \ell} \prod_{i \ne j} ((1 - \theta / M)^{|B_i|} + \varepsilon) \\
        \le & \dfrac{1}{M} \sum_{\frac{|X|}{\ell} \cdot \frac{\theta}{M} > \ell^{-0.1} \wedge (1 - \theta / M)^{|B_S|} \le \varepsilon \cdot \ell} \prod_{i = 1}^S (\varepsilon \cdot \ell + \varepsilon) \\
        \le & \dfrac{1}{M} \sum_{\frac{|X|}{\ell} \cdot \frac{\theta}{M} > \ell^{-0.1} \wedge (1 - \theta / M)^{|B_S|} \le \varepsilon \cdot \ell} \prod_{i = 1}^S 2^{2 - (C_s - 1) \cdot t} \\
        \le & 2^{2 C' \cdot \log \log N - C' \cdot (C_s - 1) \cdot \log N} = N^{-O(1)}.
    \end{align*}
\end{enumerate}

Thus, we bound the total multiplicative error as
$$
\dfrac{1}{|X|^{3 C - 1}} + 2^{-0.001 C_s \cdot t} + |X| \cdot \left( \exp(-\ell^{0.8}) + N^{-O(1)} \right) < 2^{-2 C \cdot t}.
$$
\section{Proof of \Cref{lem:allocation_k_min_wise}}\label{sec:proof_allocation_k_min_wise}

Here we complete the proof of \Cref{lem:allocation_k_min_wise}. We reproduce the lemma below for easy reference.

\begin{lemma}
    For the allocation function $g$, let $B_i := \{x \in X \setminus Y : g(x) = i\}$. Particularly, define $B_J := \bigcup_{i = 1}^{k'} B_{j_i}$. Then $g$ guarantees that:
    \begin{enumerate}
        \item When $|X| \le \ell^{0.9}$, with probability $1 - \frac{1}{\ell^{3 C} \cdot |X|^k}$, $|B_i| \le C_g + 10 \cdot \frac{k \log |X|}{\log N / \log \log N}$ for all $i \in [\ell]$ and $|B_J| \le C_g \cdot k$.

        \item When $|X| \in (\ell^{0.9}, \ell^{1.1})$, with probability $1 - 1 / \ell^{3 C \cdot k}$, the max-load $\max_{i \in [\ell]} |B_i| \le 2 \ell^{0.1}$.
        
        \item When $|X| \ge \ell^{1.1}$, with probability $1-|X|^{-3 C \cdot k}$, all buckets satisfy $|B_i| = (1 \pm 0.1) \cdot |X| / \ell$.
    \end{enumerate}
\end{lemma}

\begin{proof}
    For convenience, set $r := |X \setminus Y|$ in this proof. We prove these three cases separately.
    
    When $|X| \le \ell^{0.9}$, let $v := C_g + 10 \cdot \frac{k \cdot \log |X|}{\log N / \log \log N}$. Note that $v < C_g \cdot k$, and there comes
    $$
    \Pr_{g : \text{$C_g k$-wise}}[|B_i| \ge v] \le \binom{r}{v} \cdot (1 / \ell)^v \le (r / \ell)^v \le 1 / \ell^{0.1 v} \le \frac{1}{\ell^{4 C} \cdot |X|^k},\ \forall i \in [\ell].
    $$

    \noindent Moreover, for $B_J$, note that $k \ll \ell$. Thus it follows that
    \begin{align*}
        \Pr_{g : \text{$C_g k$-wise}}[|B_J| \ge C_g \cdot k] & \le \binom{r}{C_g \cdot k} \cdot \left( \frac{k'}{\ell} \right)^{C_g \cdot k} \le \left( \frac{r \cdot k}{\ell} \right)^{C_g \cdot k} \\
         & \le \left( \frac{k}{\ell} \right)^{0.1 C_g \cdot k} \le \left( \frac{1}{\ell} \right)^{0.5 C_g \cdot k} \le \ell^{-5 C \cdot k} \le \frac{1}{\ell^{4 C} \cdot |X|^k}.
    \end{align*}

    \noindent By a union bound, with probability $1 - \frac{1}{\ell^{3 C} \cdot |X|^k}$, $|B_i| \le v$ for all buckets as well as $|B_J| \le C_g \cdot k$.
    
    When $|X| > \ell^{0.9}$, the proof follows exactly the same logic as \Cref{lem:allocation_small_error}. Here we only show the analysis for $|X| \in (\ell^{0.9}, \ell^{1.1})$. Fix $i \in [\ell]$ and define $Z_x := \mathbf{1}(g(x) = i)$. Then $\E_{g : \text{$C_g k$-wise}}[(|B_i| - \E[|B_i|])^{C_g \cdot k}] \le O(C_g \cdot k \cdot r / \ell)^{C_g \cdot k / 2}$. Because $k \ll \ell$, we have
    $$
    \Pr_{g : \text{$C_g k$-wise}}[|B_i| - \E[|B_i|] \ge \ell^{0.1}] \le \frac{O(C_g \cdot k \cdot r / \ell)^{C_g \cdot k / 2}}{\ell^{0.1 C_g \cdot k}} \le \frac{O_{C_g}(k^{C_g \cdot k / 2})}{\ell^{0.05 C_g \cdot k}} \le \frac{1}{\ell^{0.01 C_g \cdot k}} \le \ell^{-(3 C + 1) \cdot k}.
    $$

    \noindent We obtain that with probability $1 - 1 / \ell^{3 C \cdot k}$, $\max_{i \in [\ell]} |B_i| \le 2 \ell^{0.1}$ after a union bound.
\end{proof}

\end{document}